%% file: main.tex
\documentclass[11pt]{article}
\usepackage{fullpage}
\usepackage{graphicx}

\usepackage{amsthm}
\usepackage{amsmath}
\usepackage{amssymb}
\usepackage{nicefrac}
\usepackage{xcolor}
\definecolor{Darkblue}{rgb}{0,0,0.4}
\definecolor{Brown}{cmyk}{0,0.61,1.,0.60}
\definecolor{Purple}{cmyk}{0.45,0.86,0,0}
\definecolor{Darkgreen}{rgb}{0.133,0.543,0.133}
\usepackage[colorlinks,linkcolor=Darkblue,filecolor=blue,citecolor=blue,urlcolor=Darkblue,pagebackref]{hyperref}
\usepackage[nameinlink]{cleveref}
\usepackage{aliascnt}
\usepackage{color}
\usepackage{thm-restate}
\usepackage{appendix}
\usepackage{paralist}
\usepackage{multicol}
\usepackage{lineno}

\input{macros}

\begin{document}

\title{Near-Resolution of the Tradeoff Conjecture in \\Distributed Proof Labeling Schemes} 
\date{}
\author{
    Arnold Filtser \\
	\small Bar-Ilan University \\
	\small arnold.filtser@biu.ac.il \\
	\and
	Orr Fischer \\
	\small Bar-Ilan University \\
	\small orr.fischer@biu.ac.il \\}
\maketitle

\input{abstract}

        \thispagestyle{empty}
	\newpage
	\pagenumbering{arabic}

\thispagestyle{empty}
\clearpage
\setcounter{page}{1}


\section{Introduction}
In distributed proof labeling schemes (the PLS model), our goal is to verify that a network of nodes satisfies a given property $P$. The PLS model was first introduced by \cite{KKP10}, and has been studied in the context of distributed verification, with the aim of locally certifying that a network is in a valid state while storing compact labels on the nodes. Beyond verification, the PLS model has direct applications in the design of self-stabilizing protocols and in the study of local algorithms (See \cite{F21} for a survey). 

Formally, for a graph $G = (V,E)$ with some auxiliary inputs $I$ (e.g. edge weights), a proof labeling scheme certifying a property $P$ of $(G,I)$ consists of a pair $\Pi = (\Prov,\Ver)$, where $\Prov$ is a prover strategy that assigns a label to each node, and $\Ver$ is a verification algorithm executed by every node following a label assignment. Each node $v$ decides to accept or reject based on its label and the label of its neighbors. The scheme must satisfy two conditions: 
\begin{itemize}
\item \emph{Completeness:} if $(G,I)$ satisfies $P$, all nodes accept given the labeling of the prover $\Prov$. 
\item \emph{Soundness:} if $(G,I)$ does not satisfy $P$, then for every labeling at least one node rejects. 
\end{itemize} 
The cost of a scheme $\Pi$ is the maximum label size assigned by the prover $\Prov$, as a function of the graph size $|V(G)| = n$. 

In the generalized $t$-PLS model, nodes may decide to accept or reject based on their $t$-hop neighborhood and the labels of the nodes in this neighborhood, instead of just their immediate neighborhood. This extension is motivated by the study of space-time tradeoffs, as it allows the verifier to use $t$ rounds of communication (instead of one) to potentially reduce the label sizes. In this model, many tradeoffs have been studied for specific problems, such as acyclicity \cite{O17}, MST and shortest path problems \cite{FFHPP21}, colorings \cite{AMCFNPJ23,FKMSE24}, and dominating sets \cite{FKMSE24}. On the other hand, some works investigated space-time tradeoffs of any property $P$ \cite{O17,FFHPP21,FOS21,BFZ25}. The focus of our work is on this latter category.

The $t$-PLS model was first defined by
Ostrovsky, Perry, and Rosenbaum \cite{O17}, who showed that any property $P$ can be certified by a $t$-PLS with cost $O(\min(n^2,m\log{n})/t)$. Following this work, Feuilloley, Fraigniaud, Hirvonen, Paz, and Perry \cite{FFHPP21} showed that for trees, cycles, and grids, any $1$-PLS for a predicate $P$ with cost $p$ can be transformed into a $t$-PLS for $P$ with cost $O(\lceil p/t \rceil)$. They raised the question of whether this tradeoff can be obtained in general, in a conjecture later known as the Tradeoff Conjecture. The conjecture was raised again as an open question in \cite{FOS21,AMCFNPJ23,FKMSE24,BFZ25} and the survey \cite{F21}.

\begin{conjecture}[The Tradeoff Conjecture \cite{FFHPP21}] If there exists a $1$-PLS for a predicate $P$ on all graphs with cost $p=p(n)$, then for any $t \geq 1$, there exists a $t$-PLS for $P$ with cost $O(\lceil p/t\rceil)$.
\end{conjecture}
We resolve the weak version of the conjecture (mentioned as Problem 1 in \cite{FFHPP21}).
Specifically, we prove that the Tradeoff Conjecture holds, up to a single \emph{multiplicative} $\log{n}$ factor. To our knowledge, this is the first tradeoff for any $P$ and for all graphs depending on the cost of the $1$-PLS. In the following, a configuration family $(\cG,\cI)$ is a pair consisting of a graph family $\cG$ and a set of possible auxiliary inputs $\cI$. 

\begin{restatable}{theorem}{thmgeneral}
\label{thm:general_main}
    Let $(\cG,\cI)$ be a configuration family, and $P$ a predicate on $(\cG,\cI)$. For $G \in \cG$ with $|V(G)| = n$, if there exists a $1$-PLS for $P$ with cost $p$, then for any $t \geq 1$ there exists an $O(t\log{n})$-PLS for $P$ with cost $O(\lceil p/t \rceil)$.  
\end{restatable}

Consider the problem of MST verification in graphs with integer polynomial edge weights.
Kutten and Peleg \cite{KK07} constructed a $1$-PLS with cost $O(\log^2{n})$.
Later, Feuilloley\etal \cite{FFHPP21} constructed a $t$-PLS with cost $O(\lceil \log^2{n}/t \rceil)$. However, their $t$-PLS can only be applied for $t=O(\log n)$. Thus, the scheme of \cite{FFHPP21} has cost $\Omega(\log n)$ even for very large $t$.
By applying our black-box \Cref{thm:general_main} combined with \cite{KK07}, we get a general tradeoff, and an improvement over \cite{FFHPP21} for the parameter range $t\ge \log^2 n$.
In particular, for $t=\Omega(\log^3 n)$, we get constant cost.

\begin{corollary}[$t$-PLS for MST verification]
    For any $t \geq \log{n}$, there exists a $t$-PLS for the MST verification problem with cost $O(\lceil \log^3{n}/t \rceil)$.
\end{corollary}

Next, we prove that the Tradeoff Conjecture holds in fixed minor-free graphs, up to a single \emph{additive} $\log{n}$ factor. Previously, Fischer, Oshman and Shamir \cite{FOS21} showed a bound of $\widetilde{O}(\lceil \Delta p/\sqrt{t} \rceil)$ on the cost, where $\Delta$ is the maximum degree in the graph. Our result removes the dependence in $\Delta$, and improves the dependence on $t$ quadratically. Many well-known graph families are minor free, such as planar graphs, as well as graphs of bounded genus or treewidth. We state our results for $K_r$-minor free graphs, which are in particular $H$-minor free for any $H$ such that $|V(H)| \leq r$.

\begin{restatable}{theorem}{thmminorfree}
\label{thm:main_minor_free}
    Let $r \geq 1$ be a constant, and let $(\cG,\cI)$ be a configuration family where $\cG$ is the family of all $K_r$-minor free graphs, and let $P$ be a predicate on $(\cG,\cI)$. If there exists a $1$-PLS for $P$ in $(\cG,\cI)$ with cost $p$, then for any $t \geq 1$, there exists a $t$-PLS for $P$ in $(\cG,\cI)$ with cost $O(\lceil p/t\rceil+\log{n})$.  
\end{restatable}

In \cite{FFHPP21}, the authors additionally asked whether the tradeoff might scale linearly with the ball growth of the graph $b(t) = \min_v |B_t(v)|$,\footnote{$B_t(v)=\{u\in V\mid\dist_G(v,u)\le t\}$ denotes the ball around $v$ of radius $t$ w.r.t. the shortest path metric of $G$.} although they considered this to be most likely false.

\begin{question}[\cite{FFHPP21} Open Problem 2] Assuming there exists a $1$-PLS for a predicate $P$ on graphs with cost $p = p(n)$, does there exist for any $t \geq 1$ a $t$-PLS for $P$ with cost $\widetilde{O}(\lceil p/b(t)\rceil)$?
\end{question}

We confirm that this does not hold, and show that even if all $t$-neighborhoods are of size $\Omega(n)$, the best scaling possible is $\Omega(\lceil p/t \rceil)$. Moreover, this phenomenon already occurs in planar graphs.

\begin{restatable}{theorem}{thmdisproof}
\label{thm:disproof_strong}
    There exists a predicate $P$ on an infinite family of planar graphs $\cG = \{G_{t,m}\}_{t,m = 1}^{\infty}$ and on some auxiliary input set $\cI$, where for $G_{t,m} \in \cG$ we have $\min_{v \in V(G_{t,m})}|B_{t}(v)| = \Omega(|V(G_{t,m})|)$, and there is a $1$-PLS for $P$ with cost $O(m)$, but any $t$-PLS for $P$ has cost $\Omega(m/t)$.
\end{restatable}

\paragraph*{Additional related works:} The PLS model (i.e. $1$-PLS model) has been extensively studied in a broad range of topics, such as for various graph properties \cite{KK07,KKP10,GS16,CPP20}, for properties expressible in some logic frameworks \cite{FBP22,FMRT24,CKM25}, and for memberships in graph classes \cite{FFMRRT21,BFP24}. Given the very rich literature on PLS, we refer to the survey on local certification and the PLS model of Feuilloley \cite{F21}, as well as the surveys on related topics of Feuilloley and Fraigniaud \cite{FF16}, and of Suomela \cite{S13}.

\subsection{Technical Overview}
We begin by presenting a key framework for obtaining tradeoff theorems, used both in prior works and in our work, which shows how the existence of certain graph decompositions directly yields tradeoff theorems. We then introduce a new type of decomposition, called \emph{Two-Separated partitions} (TS partitions), which extends the prior decompositions while implying similar tradeoff theorems.

Following this, we show three constructions of TS partitions in graphs: \begin{inparaenum}[(a)]
    \item a warmup construction which yields a weaker tradeoff theorem for general graphs (cf. \Cref{thm:general_main}), but its ideas are used throughout the paper, \item constructing TS partitions from padded decompositions (introduced later), which implies our results for minor free graphs (\Cref{thm:main_minor_free}), and
    \item a construction which yields our tradeoff theorem for general graphs (\Cref{thm:general_main}).
\end{inparaenum} Finally, we discuss the techniques for refuting the stronger variant of the Tradeoff Conjecture (\Cref{thm:disproof_strong}).

\paragraph*{Tradeoff from partition-based decompositions:}
To obtain Theorems~\ref{thm:general_main} and \ref{thm:main_minor_free}, we use a variation on a partition-based framework previously used in \cite{FFHPP21,FOS21}. Before introducing our variation, we first informally describe this framework as it appears in prior works. 

The key combinatorial object needed for this framework is the existence of a decomposition of the graph into connected clusters $C_1, \dots, C_k$, each of diameter at most 
$\alpha\cdot t$ (for some small enough constant $\alpha\le 1$), such that for the set of boundary vertices\footnote{The boundary $X_i$ of a cluster $C_i$ consists of all the vertices in $C_i$ with neighbors out of $C_i$. Formally, $X_i = C_i \cap B_1(V \setminus C_i)$, where $B_1(S)$ is the ball of radius 1 around $S$.} $X_i$ of $C_i$ we have that $\forall_i |X_i|/|C_i| \leq \varepsilon$ for some small $\varepsilon = \varepsilon(n,t)$.
Given such a decomposition, we describe how to construct a $t$-PLS $\Pi_t$ from a $1$-PLS $\Pi_1$, while paying only roughly an $\varepsilon$-fraction of the cost of $\Pi_1$. Specifically, if $\Pi_1$ costs $p$, we need to construct a labeling for every graph $G$ satisfying $P$, such that the maximum label size is roughly $O(\lceil \varepsilon  p \rceil)$, and describe a verifier function that accepts this labeling, but rejects all labelings on graphs not satisfying $P$.

\vspace{2mm}
\noindent \emph{Label construction for $\Pi_t$:} Recall that $G$ satisfies $P$ if and only if there is some labeling $\ell_V$ for $\Pi_1$, such that all nodes accept according to $\Pi_1$. We describe a $t$-PLS $\Pi_t$ that certifies $P$ by proving to the nodes that such a labeling $\ell_V$ exists. Assume that a $1$-PLS $\ell_V$ exists with cost $p$. We create the $t$-PLS labels of the vertices in  cluster $C_i$ by taking the $1$-PLS labels of its boundary vertices $\{\ell_V(x)\}_{x\in X_i}$, and splitting them among all the cluster vertices (say lexicographically). Additionally, each vertex $v\in C_i$ will get an \emph{auxiliary label} containing additional information so that $v$ can identify its cluster $C_i$.
We denote the maximum possible size of the auxiliary label by $c$, and call it \emph{description cost}. The resulting label size of $C_i$ vertices will be $O\left(\left\lceil \frac{|X_{i}|}{|C_{i}|}\cdot p\right\rceil +c\right)=O\left(\left\lceil \eps p\right\rceil+c \right)$.
Thus the cost of $\Pi_t$ is $O\left(\left\lceil \eps p\right\rceil+c \right)$.

\vspace{2mm}
\noindent \emph{Verification under $\Pi_t$:} 
For each cluster $C_i$, a leader $v_i\in C_i$ is designated.
Given the labels of $C_i$, $v_i$ can reconstruct the labels $\ell_{V}(x)$ for all boundary vertices $x \in X_i$.
Further, let $\bar{C}_i$ be the collection of clusters with edges also incident to $C_i$ (i.e. neighboring clusters, excluding $C_i$), and by $\bar{X}_i=\cup_{C_j\in \bar{C}_i}X_j$ the union of their boundary vertices.
Note that for each cluster $C_j\in\bar{C}_i$, as all $C_j$ vertices are at distance $\le 2\alpha\cdot t$ from $v_i$, $v_i$ can reconstruct the cluster $C_j$, and the $\Pi_1$-labels of its boundary vertices $X_j$. 
In other words, $v_i$ can compute $\{\ell_V(x)\}_{x\in\bar{X}_i}$.
We observe the following:
\begin{itemize}
    \item For every $v\in C_i\setminus X_i$, all its neighbors are in $C_i$.
    \item For every $x\in  X_i$, all its neighbors are in $C_i\cup\bar{X}_i$.
\end{itemize}
All $C_i$ vertices other than the leader $v_i$ will accept regardless of the input.
$v_i$ will accept iff there are labels $\{\ell_V(v)\}_{v\in C_i\setminus X_i}$ for the non-boundary vertices so that all $C_i$ vertices (including $X_i$) will be accepted by the $\Pi_1$-verifier simultaneously (with accordance to the given labels of $\bar{X}_i\cup X_i$).
Showing completeness and soundness is now straightforward. 
If $(G,I)$ satisfies $P$, and a $\Pi_1$ labels indeed exists, each leader $v_i$ can use them to convince itself that there are satisfying labels and thus will accept.
For soundness, one can use the fact that there are no edges between yet undetermined (i.e. non-boundary) vertices in different clusters, and thus if each leader $v_i$ accepts separately, the combined labels they computed will satisfy the $\Pi_1$-verifier, and thus by the soundness of $\Pi_1$, $(G,I)$ satisfies $P$.

\vspace{2mm}
\noindent \emph{Challenge:} 
As it turns out, constructing decompositions such that the fraction of boundary vertices is small in all the clusters simultaneously is quite challenging. In the following, we describe how we can relax the requirements from the decomposition to get considerably better results.

\paragraph*{New combinatorial notions - Two-Separated partitions and cluster-degeneracy:}

In the previous part, we discussed how the existence of a decomposition into connected clusters $C_1,\dots,C_k$, with diameter $O(t)$ and boundary-to-size ratio of $|X_i|/|C_i| \leq \varepsilon$ directly implies a tradeoff theorem in which the $t$-PLS costs an $\varepsilon$-fraction of the original $1$-PLS. As our first contribution, we present the notion of a Two-Separated partition (TS partition), which generalizes the decomposition above while still implying the same tradeoff guarantees.

A $(t,\varepsilon)$-\emph{Two-Separated partition} is a pair $(\cC,X)$ where  $X \subseteq V$ is called the \emph{separating set}, and $\cC$ is a partition of $V$, satisfying
\begin{inparaenum}[(a)]
    \item for $i \ne j$, any path between two nodes $u \in C_i \setminus X$ and $v \in C_j \setminus X$ must contain two consecutive vertices in $X$, 
    \item every cluster in $\cC$ has weak diameter at most $t$,\footnote{The weak diameter of a cluster $C$ is $\max_{u,v \in C} \dist_G(u,v)$ the maximum pairwise distance w.r.t. the original distances. In particular, a cluster may be disconnected.} and
    \item $\max_i |C_i \cap X|/|C_i| \le\varepsilon$. (See \Cref{fig:TSandClusterDeg}(a) for an example)
\end{inparaenum}

Conceptually, this new decomposition allows the set $X_i = C_i \cap X$ to differ from the boundary of the cluster $C_i$. This allows, for example, for a cluster to be disconnected, as long as it has a low weak-diameter, and also allows a more flexible manner to balance how much ``boundary information'' is encoded on each cluster. Using essentially the same technique as the partition-based schemes, we show that the existence of a $(t,\varepsilon)$-TS partition implies that a $1$-PLS with cost $p$ can be turned into an ``equivalent'' $O(t)$-PLS with cost $O(\lceil \varepsilon p\rceil+c)$ for some description cost $c$.

Our main method in the paper of constructing $(t,\varepsilon)$-TS partitions is through the notion of $\varepsilon$-cluster-degeneracy, introduced next. A partition of $V$ into an ordered set of clusters $\cC = \{C_1,\dots,C_k\}$ has $\varepsilon$-cluster-degeneracy if  
\[\forall i\in[k],\quad
\left|\left\{ v\in C_{i}\mid \dist_{G}\left(v,V\setminus\cup_{j=1}^{i}C_{j}\right)\le2\right\} \right|\le\varepsilon\cdot|C_{i}|.
\]
See \Cref{fig:TSandClusterDeg}(b) for a pictorial example. We show that a decomposition into clusters with $\varepsilon$-cluster-degeneracy and weak diameter at most $t$ implies the existence of a $(t,\varepsilon)$-TS partition (\Cref{lem:ds_implies_ts}). As stated above, the existence of these decompositions directly implies a tradeoff theorem, 
where the cost is $O(\lceil\eps p\rceil)$ $+$ the description cost of the TS partition.
Thus our primary effort is aimed at showing the existence of TS partitions with small description cost. In the following parts we present three constructions of TS partitions, for both general and minor free graphs.

\begin{figure}[ht]
  \centering
  \begin{minipage}{0.5\textwidth}
    \centering     \includegraphics[width=\linewidth]{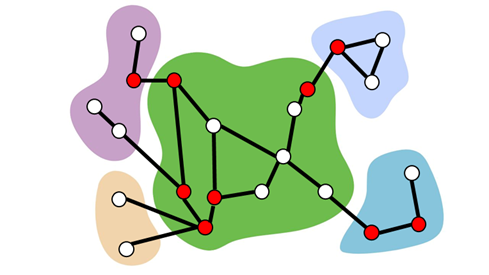}
  \end{minipage}\hfill
  \begin{minipage}{0.50\textwidth}
    \centering
     \includegraphics[width=\linewidth]{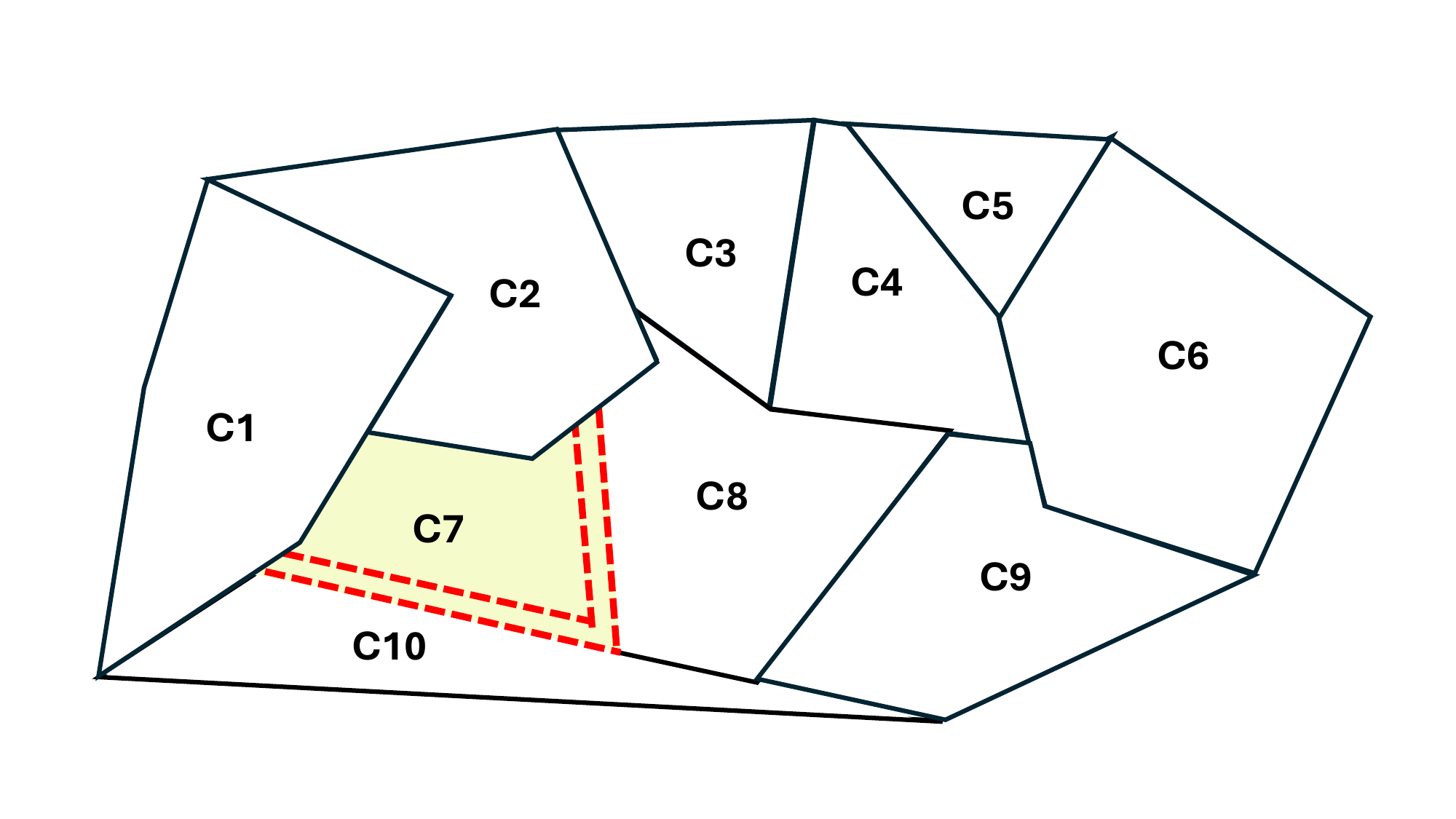}   
  \end{minipage}\hfill
  \caption{(a) Illustration of a $(5,2/3)$-TS partition. Nodes of $X$ marked in red. All clusters have weak diameter $\leq 5$, where some clusters are internally disconnected. Worst red nodes to size ratio is in the bottom-right cluster, with $\varepsilon = 2/3$. All white-to-white paths between clusters must pass two consecutive red nodes.  (b) Illustration of cluster-degeneracy. In the example, we highlight for $C_7$ its $2$-boundary with the clusters that follow it. We remark that in general, clusters can be disconnected.}
  \label{fig:TSandClusterDeg}
\end{figure}

\paragraph*{Warmup - a simple solution for general graphs:} As a warmup, we consider a simple ball carving process that outputs for any graph a decomposition into clusters with $(1/t)$-cluster-degeneracy and weak diameter $O(t\log{n})$. That is, we show that $V$ can be partitioned into clusters $C_1,\dots,C_k$, each of weak diameter $O(t\log{n})$, such that 
\[\forall i \in [k],\quad
\left|\left\{ v\in C_{i}\mid \dist_{G}\left(v,V\setminus\cup_{j=1}^{i}C_{j}\right)\le2\right\} \right|\le |C_{i}|/t.
\]

The process works in iterations. For iteration $i \geq 1$, assume we are given some set of alive nodes $\alive_i \neq \emptyset$, initialized as $\alive_1 = V$. We take an arbitrary node $v_i \in \alive_i$, and consider balls of growing radius centered at $v_i$. We say that a ball $B_j(v_i)$ of radius $j \geq 2$ is \emph{expanding} if 
$|B_{j}(v_i)\cap \alive_i|/|B_{j-2}(v_i)\cap \alive_i| > 1+1/t$. We note that there exists a non-expanding ball centered at $v_i$ of radius $O(t\log{n})$. Otherwise, due to the many expanding layers, we have $|B_{O(t\log{n})}(v_i)| \geq (1+1/t)^{\Theta(t\log{n})}  > n$.\footnote{We remark that if $B_j(v_i)$ contains all alive nodes, then $B_{j+2}(v_i)$ is a non-expanding ball.} If we find a non-expanding ball $B_j(v_i)$, we define the cluster $C_i = B_j(v_i)$ and set $\alive_{i+1} = \alive_i \setminus C_i$, terminating when $\alive_{i+1} = \emptyset$.

Using simple calculations, it is easy to see that the resulting clusters are $(1/t)$-cluster-degenerate: Each cluster is a non-expanding ball, whose cluster's size is $|B_r(v_i) \cap \alive_i|$, and its $2$-hop boundary with subsequent clusters is $|(B_{r}(v_i) \setminus B_{r-2}(v_i)) \cap \alive_i|$. Therefore by definition of a non-expanding radius, we have
\[|(B_{r}(v_i) \setminus B_{r-2}(v_i)) \cap \alive_i| = |B_{r}(v_i) \cap \alive_i| - |B_{r-2}(v_i) \cap \alive_i| \leq (1/t) \cdot|B_r(v_i) \cap \alive_i| = |C_i|/t~.\]

We conclude that a $(O(t\log{n}),1/t)$-TS partition exists in any graph. Furthermore, using standard arguments, one can show that any TS partition can be described with a cost of $O(\log{n})$ bits (See \Cref{lem:tsminorcompute}). This immediately implies that any $1$-PLS for a property $P$ with cost $p$ can be transformed to an $O(t\log{n})$-PLS for $P$ with cost $O(\lceil p/t \rceil + \log{n})$. 

However, many properties of interest (e.g. computing a spanning tree \cite{KKP10,GS16}, MST \cite{KK07}, or certifying planarity \cite{FFMRRT21}) have a $1$-PLS with logarithmic or close to logarithmic cost. Ideally we would like a $t$-PLS for any property, that converges to constant cost as $t$ grows. We could obtain this if we could reduce the additive factor in the tradeoff to be bounded by a constant. Therefore, a natural question to ask is whether there exists a TS partition with the parameters above, but with constant description cost? Another natural question is whether we can remove the multiplicative dependency of the tradeoff in $n$, at least for some large graph class? We positively answer the two questions, as discussed next.

\paragraph*{$K_r$-Minor free graphs and padded decompositions:} In \Cref{sec:minor_free_graphs}, we show that for the class of $K_r$-minor free graphs (for constant $r \geq 1$), any $1$-PLS for a property $P$ with cost $p$ can be transformed into a $t$-PLS for $P$ with cost $O(\lceil p/t \rceil + \log{n})$, which is optimal up to a single additive logarithmic factor.\footnote{See e.g. \Cref{thm:disproof_strong} for an $\Omega(1/t)$ lower bound on the tradeoff of some problems, though there are many such bounds in the literature.} The main tool we use is the existence of good ``padded decompositions'' in these graphs.

See \Cref{def:PadDecompostion} for a formal definition of padded decomposition. 
Here it will suffice to note that if a graph has padding parameter $\beta > 0$, then for any $t = \Omega(1)$, there is a \emph{randomized} partition of $V$ into clusters, such that \begin{inparaenum}[(a)]
    \item each cluster has weak diameter at most $t$, and
    \item for every node $v \in V$, 
$\Pr(B_2(v) \nsubseteq C_v) \leq 2\beta/t$,
where $C_v$ is the unique cluster containing $v$, and $B_2(v)$ is the ball of radius $2$ centered at $v$.\footnote{Recall that the partition is randomized, meaning the randomness is taken over its distribution.}
\end{inparaenum} The padding parameter $\beta$ of different graphs is a widely studied topic that has many algorithmic applications (See \cite{CF25} and the references therein). It is known that for $K_r$-minor free graphs $\beta = \Theta(\log{r})$ \cite{CF25} ($O(1)$ if $r$ is constant), while in general graphs $\beta = \Theta(\log{n})$ \cite{Bar96}. 

As we next show, the existence of padded decompositions with low padding parameter $\beta$ in a graph $G$ implies the existence of a good TS partition in $G$. We start by showing that graphs with padding parameter $\beta$ have some cluster $C \subseteq V$ with weak diameter at most $t$, and $|C \cap B_2(V\setminus C)| \leq (2\beta/t)|C|$. Let us sample a random padded decomposition with padding parameter $\beta$ and weak diameter $t$, i.e. a randomized partition into clusters $\{C_1,\dots,C_k\}$, such that the weak diameter of the clusters is at most $t$, and that property (b) holds. Let $\Gamma_2$ be the $2$-boundary of the clusters, i.e. $\Gamma_2 = \{v \mid  \dist_G(v, V\setminus C_v) \leq 2\}$. We notice that property (b) of the decomposition is equivalent to:
\[\Pr(v \in \Gamma_2) = \Pr(B_2(v) \nsubseteq C_v) \leq 2\beta/t.\]

Therefore, by linearity of expectation, $\mathbb{E}(|\Gamma_2|) \leq (2\beta/t)\cdot n$. We notice that 
\[\min_{i \in [k]} |C_i \cap \Gamma_2|/ |C_i| \leq \left(\sum_{i=1}^k |C_i \cap \Gamma_2|\right)/\left(\sum_{i=1}^k |C_i|\right) = |\Gamma_2|/n,\]
where the first inequality follows from the mediant inequality.\footnote{The mediant inequality states that if $a_1,b_1,\dots,a_\ell,b_\ell >0$ are real numbers, then $\min_{i \in [\ell]}\frac{a_i}{b_i} \leq \frac{\sum_{i=1}^\ell a_i}{\sum_{i=1}^\ell b_i}$.} In expectation, we have  \[\mathbb{E}(\min_{i \in [k]} |C_i \cap \Gamma_2|/ |C_i|) \leq \mathbb{E}(|\Gamma_2|)/n \leq 2\beta/t.\] 
We conclude that by the probabilistic method, there exists a cluster $C$ such that $|C \cap B_2(V \setminus C)| \leq (2\beta/t) |C|$. Moreover, using the same argument one can show that for any set $\alive \neq \emptyset$, we can find a cluster $C  \subseteq \alive$ with weak diameter at most $t$ and $|C \cap B_2(\alive \setminus C)| \leq (2\beta/t) |C \cap \alive|$.

Given the existence of such a cluster $C$, we can perform a similar ball carving process as in the warmup. We define a process with iterations where in iteration $i \geq 1$ we are given some set of alive nodes $\alive_i \neq \emptyset$, initialized as $\alive_1 = V$. Using the padded decomposition, we argue that there exists some cluster $C_i$ such that $|C_i \cap B_2(\alive \setminus C_i)| \leq (2\beta/t)|C_i|$. We add $C_i$ to our partition, and set $\alive_{i+1} = \alive_i \setminus C_i$, terminating if $\alive_{i+1} = \emptyset$. The resulting partition can be shown to be a partition into clusters with weak diameter $t$, and $(2\beta/t)$-cluster-degeneracy, which for $K_r$-minor free graphs translates to $O(1/t)$-cluster-degeneracy, and implies \Cref{thm:main_minor_free}.

Interestingly, by using the optimal padded decomposition of general graphs \cite{Bar96}, we also obtain an alternative method to the same result as in the warmup - that a $1$-PLS for a property $P$ with cost $p$ can be transformed to an $O(t\log{n})$-PLS for $P$ with cost $O(\lceil p/t \rceil + \log{n})$.

\paragraph*{General Graphs, with Constant Description Cost:} In \Cref{sec:general_graphs},  we show that for the class of all graphs, any $1$-PLS for a property $P$ with cost $p$ can be transformed into an $O(t\log{n})$-PLS with cost $O(\lceil p/t \rceil)$. This cost's additive factor is bounded by a constant, hence this scheme converges to a constant number of bits as $t$ grows. Our goal is to adapt the construction of the warmup to obtain an $(O(t\log{n}),1/t)$-TS partition that can be locally described using $O(1)$ bits per node.

\vspace{2mm}
\noindent\emph{Attempt 1 - Ball carving with advice:} We design a compact description for the warmup construction. Recall that in this process we took in each step some living node $v$, carved a ball $B_r(v)$ around it for some $r$, added it to our clusters, and removed $B_r(v)$ from the set of living nodes. Let us further assume that we iterate over the vertices in increasing order of id, i.e. in each step we take the node with the smallest identifier that is still alive.

We describe a label of size $O(\log{t}+\log\log{n})$ which allows nodes to recognize their cluster, given their labeled $O(t\log{n})$-neighborhood: we mark for each node $v$ if a ball $B_{r(v)}(v)$ was taken as a cluster, and if so, its radius $r(v)$. We note that a node $u$ can locally simulate the process, and deduce its cluster, by finding the node $v$ with the smallest id in its $O(t\log{n})$ neighborhood such that $u \in B_{r(v)}(v)$, and assigning itself to a cluster with identifier $\ID(v)$. Similarly, node $u$ can repeat this process for each node in its surrounding, deducing the entire cluster. Given an additional label for marking nodes in the separating set $X$, a node can verify that all the TS partition properties hold. 

The only real cost of this scheme is describing the radius for each node. One natural approach to reduce the cost is to encode this information across each cluster, however there is no guarantee the clusters are sufficiently large to ensure it can be stored by constant size labels. We instead take a different approach: We observe that sharing some common string $S$ across the graph is much cheaper than assigning each node $v$ some distinct string $S_v$ of its own (See \cite{O17,FFHPP21,FOS21}). Therefore, instead of describing the radii explicitly, we would like some shared string from which all nodes could deduce their encoded radius, as well as the radii of nodes in their $O(t\log{n})$-neighborhood.

As our key observation we show that picking a \emph{random radius} for each node is good, assuming additional alterations to the process. Therefore, assuming all nodes have access to a long shared randomness string, they can deduce the radius of every node. We then condense this string to $O(\log{n})$ bits using randomness reduction techniques, and have the prover share a good random string for our graph $G$ across all nodes. We describe the new process, and explain how to identify clusters.

\vspace{2mm}
\noindent\emph{Solution - Ball carving with randomized radius:}
 Assume we have a given random string $R$. The new process runs in iterations $i = 1,\dots,n$, where we iterate over the vertices in lexicographic order while maintaining the set of currently living (unclustered) nodes $\alive_i$, where we iterate over $v_i$ even if $v_i \notin \alive_i$. In step $i$, we check whether $v_i$ has a living node within distance $2t\log{n}$. If it does, it samples a random even radius $r_i\in [2t\log{n}+2,8t\log{n}]$ using $R$, and considers the candidate cluster $C_i=B_{r_i}(v_i)\cap\alive_i$ (which does not contain $v_i$ if $v_i \notin \alive_i$). As before, we take the cluster $C_i = B_{r_i}(v_i) \cap \alive_i$ if it does not expand, i.e. if
\[|B_{r_i}(v_i) \cap \alive_i|/|B_{r_i-2}(v_i) \cap \alive_i| \leq 1+1/t.\]
then $C_i$ is carved out of $\alive_i$ and its 2-boundary with subsequent clusters is added to $X$. The output is the carved clusters $\cC$ together with the marked boundary nodes $X$, which we show (with high probability) form an $(O(t\log n),1/t)$-TS partition. Intuitively, we show that in any fixed step $i$, the cluster $C_i$ is non-expanding with constant probability over the choice of $r_i$. Moreover, because we iterate over all nodes (even those already clustered), and since every node $u$ is contained in $B_{2t\log{n}}(v_i)$ for $\Omega(t\log n)$ different nodes $v_i$, we get that all nodes are clustered with high probability.

Unlike the prior attempt, all a node $u$ needs to know to identify its cluster is (a) the randomness string, (b) for every node $v$ in its vicinity, whether its cluster was taken, and (c) whether $v$ is in $X$. All this information can be shared with $v$ using $O(1)$ bits, using some randomness reduction technique, and a dedicated PLS.

\paragraph*{Refutation of the stronger variant of the Tradeoff Conjecture:} In \Cref{sec:lower}, we show that having $\Theta(n)$-sized $t$-neighborhoods is insufficient for obtaining a tradeoff better than $O(\lceil p/t \rceil)$, even for the case of planar graphs, refuting the stronger variant of the Tradeoff Conjecture as discussed above. We do so by defining our predicate $P$ as an equality problem between distant nodes, and reduce it to 2-player non-deterministic equality.   

Our construction is very simple. We design a layered graph with $2t+3$ layers, where odd layers have one node, and even layers have $m$ nodes, with the nodes of each layer connected to all nodes of the neighboring layers (See \Cref{fig:lowerbound}). The leftmost and rightmost nodes, $v^1$ and $v^{2t+3}$, receive inputs $I(v^1)$ and $I(v^{2t+3})$ from $\{0,1\}^{m^2}$, and we define the predicate $P$ to assert that $I(v^1) = I(v^{2t+3})$. We design a $1$-PLS where even layers relay and split $I(v^1)$ into $m$ parts, each part given to some node in the layer, which requires a label size $O(m)$ bits per node. We show that a $t$-PLS on $P$ with cost $p$ implies a $2$-player non-deterministic equality protocol of cost $O(p \cdot m \cdot t)$. In combination with the fact that $m^2$-bit equality costs $\Omega(m^2)$ in the 2-player model, we obtain a lower bound of $\Omega(m/t)$ on the cost of any $t$-PLS for $P$.

We remark that one can easily transform $P$ into a pure graph property (i.e. without any auxiliary input) by encoding $I(v^1),I(v^{2t+3})$ through two low-diameter subgraphs with $\Theta(m^2)$ edges connected to $v^1,v^{2t+3}$ respectively. However the graph then is not necessarily planar. For planar graphs, a similar graph property with $1$-PLS of cost $O(\log{m})$, versus $t$-PLS of cost $\Omega(\log{m}/t)$ can be easily shown by encoding two low-diameter graphs with $\Theta(m)$ edges next to each of the nodes $v^1,v^{2t+3}$.

\begin{figure}[ht]
  \centering  \includegraphics[width=0.7\textwidth]{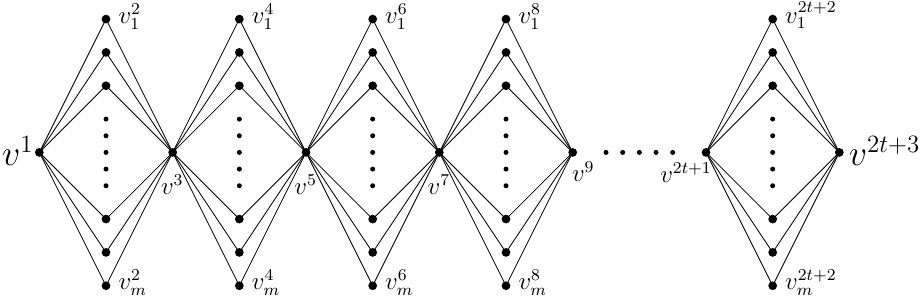}
    \caption{The graph $G_{t,m}$. It has $2t+3$ layers, and in each even layer has $m$ nodes. The leftmost and rightmost nodes of the graph $v^1,v^{2t+3}$ receive as input some strings $I(v^1),I(v^{2t+3})$, and the configuration is in $P$ if $I(v^1) = I(v^{2t+3})$.}
  \label{fig:lowerbound}
\end{figure}

\section{Preliminaries}
\label{sec:prelims}
We denote $[k] = \{1,\dots,k\}$, and $[k_1,k_2] = \{k_1,\dots,k_2\}$. Given functions $f_1,f_2$ with disjoint domains, the combined function is defined as $f = f_1 \cup f_2$, where $f(x) = f_1(x)$ if $x \in \operatorname{Dom}(f_1)$, and
$f(x) = f_2(x)$ if $x \in \operatorname{Dom}(f_2)$.

\paragraph*{Graph notation:}
For a graph $G = (V,E)$, and a pair of vertices $u,v\in V$, $\dist_G(u,v)$ denotes the shortest path metric in $G$ (that is the minimum number of edges in a $u$-$v$ path). For $V_1,V_2 \subseteq V$, let $E(V_1,V_2) = \{\{v_i,v_j\} \in E \mid v_i \in V_1,v_j \in V_2\}$ be the edges between $V_1,V_2$. $G[C] = (C,E(C,C))$ denotes the induced graph on $C$, i.e. it contains the edge set $\{\{u,v\} \in E \mid u,v \in C\}$. 
The (strong) diameter of a cluster  $C \subseteq V$ is the maximum pairwise distance in the induced graph: $\max_{u,v \in C} \dist_{G[C]}(u,v)$, while the weak diameter is the maximum pairwise distance w.r.t. original distance: $\weakdiam(C) = \max_{u,v \in C} \dist_G(u,v)$. We note that $C$ can have a finite weak diameter while not being connected. 
A partition $\cC$ of $V$ is weakly/strongly $\Lambda$-bounded if every cluster has weak/strong diameter at most $\Lambda$.
Let $N(v) \subseteq V$ denote the neighborhood of $v$. For an integer $t \geq 1$, let $B_t(v)=\{u\in V\mid \dist_G(u,v)\le t\}\subseteq V$ be the $t$-hop neighborhood of $v$ in $G$, also called the ball of radius $t$ around $v$. For a set $V' \subseteq V$, let $B_t(V') = \bigcup_{v \in V'} B_t(v)$.

A graph $H$ is said to be a \emph{minor} of $G$ if $H$ can be obtained from $G$ by a series of vertex deletions, edge deletions, and edge contractions.\footnote{Given a graph $G=(V,E)$ and an edge $e=\{u,v\}\in E$, by contracting $e$ we obtain a new graph $G'$ where $u,v$ are contracted into a new single vertex $w$, such that $w$ is adjacent to $(N(u) \cup N(v))\setminus\{u,v\}$.}
The family of $H$-minor free graphs consists of all graphs $G$ that do not contain $H$ as a minor (i.e., no series of operations leads from $G$ to $H$).

\paragraph*{String notation:}
Given a collection of strings $S_1,\dots,S_\ell$, let $(S_1,\dots,S_\ell)$ denote a single string obtained by concatenating $S_1,\dots,S_\ell$ using a fixed delimiter scheme, so that from $(S_1,\dots,S_\ell)$ one can uniquely recover the number of strings $\ell$ and each of the strings $S_1,\dots,S_\ell$. With some abuse of notation, we use $(S_1,\dots,S_\ell)$ to denote both the tuple of strings and its encoded representation as a single string; the intended meaning will always be clear from context. For an integer $p$, let $\{0,1\}^{\le p}$ denote the set of all binary strings of length at most $p$. For an integer $m$, let $\Pad(S,m)$ denote an encoding of $S$ as a string whose length is the smallest multiple of $m$ that is larger than $|S|$, such that $S$ can be uniquely recovered from $\Pad(S,m)$. We remark that both the delimiter scheme and the padding scheme can be trivially implemented by implementing a larger alphabet over the binary encoding.

\paragraph*{Proof Labeling Schemes:}
For a graph family $\cG$ and set of strings $\cI$, a $(\cG,\cI)$-\emph{configuration} is a pair $(G,I)$ where $G \in \cG$ is a graph, and $I:V \rightarrow \cI$ is an assignment of inputs to the nodes of $G$. A typical example of additional inputs $I$ that can be given to nodes is weights on incident edges, such as in the MST problem \cite{KK07}, or shortest path problems \cite{FFHPP21}.

Let $(\cG,\cI)$ be the family of all $(\cG,\cI)$-configurations, for some graph family $\cG$ and set of strings $\cI$. For a predicate $P:\cG \times \cI \rightarrow \{0,1\}$, we say that $(G,I) \in P$ if $P(G,I) = 1$.  

\begin{definition}
\label{def:labeling_scheme}
    Let $(\cG,\cI)$ be a configuration family. Then for $t \geq 1$, a $t$-hop proof labeling scheme ($t$-PLS) for a predicate $P: \cG \times \cI \rightarrow \{0,1\}$ is a pair $\Pi = (\Prov, \Ver)$ satisfying
\begin{itemize}
	\item $\Prov$, also called the honest prover, is a mapping taking as input a configuration $(G,I) \in P$, and producing a label assignment $\Prov(G,I) = ( \ell(v))_{v \in V}$ for all nodes of $G$,
		where $\ell(v) \in \{0,1\}^{*}$ for each $v \in V$.
	\item $\Ver$, also called the verifier, is a function receiving as input the size of the graph $|V(G)| = n$, and all the information in the $t$-hop neighborhood of $v$, including the $t$-hop topology of $v$ (with the corresponding identifiers), and for each node $u \in B_t(v)$, its additional inputs $I(u)$ and label assignments $\ell(u)$, and outputs a value in $\{\accept,\reject\}$ (i.e. \emph{accepts} or \emph{rejects} respectively).\footnote{Equivalently, $\Ver$ can be seen as a $t$-round deterministic distributed algorithm with unlimited bandwidth that collects all information in the $t$-neighborhood of $v$, and outputs a value in $\{\accept,\reject\}$, and additional output.} Let $\Ver(G,I,\ell,v)$ denote the output of the function at node $v \in V$, in configuration $(G,I)$, and with labels $\ell = (\ell(v))_{v \in V}$. If $\Ver$ accepts, we also allow it to output additional auxiliary values.
\end{itemize}

Additionally, the following holds:
\begin{itemize} 
	\item Completeness: for every configuration $(G, I) \in P$
		and for every node $v \in V$, we have that $\Ver(G, I, \Prov(G,I),v) = \accept$.
    \item Soundness: for every configuration $(G, I) \in (\cG,\cI)$
		and certificate assignment $\ell = ( \ell(v))_{v \in V}$,
		if $\Ver(G,I,\ell,v) = \accept$ for all $v \in V$,
		then $(G,I) \in P$.
\end{itemize}
\end{definition}

We denote the label given to a node $v$ by $\Prov(G,I)$ as $\Prov(G,I,v)$. The \emph{cost} of a $t$-PLS $(\Prov, \Ver)$ on a configuration $(G,I) \in P$ is defined as $\mathrm{cost}(\Prov,G,I) = \max_{v \in V}|\Prov(G,I,v)|$, i.e. the length of the longest label given to a node by $\Prov(G,I)$. The cost of a $t$-PLS $(\Prov, \Ver)$ on the configuration family $(\cG,\cI)$ is defined by    
\[\mathrm{cost}(\Prov, P , n) = \max{\{\mathrm{cost}(\Prov,G, I) \mid (G, I) \in P , |V(G)| = n}\}.\]


\paragraph*{Verifier knowledge and additional assumptions:} Throughout the paper, we assume that nodes have unique identifiers encoded using $O(\log{n})$ bits. We further assume that $\Ver$ is given access to the size of the graph, i.e., $|V(G)| = n$, and that the graph is connected. These latter two assumptions are only used in \Cref{sec:app_general}, and we believe they can be removed given some additional technical details, as briefly discussed at the end of that section.

\section{Two-Separated Partitions and Cluster-Degeneracy}
\label{sec:TSpartition}
In this section, we present two key notions: Two-Separated partitions (TS partitions) and $\varepsilon$-cluster-degeneracy, and relate these notions to obtaining tradeoff theorems. We start by introducing TS partitions, which are our key combinatorial objects whose existence implies the tradeoff theorems.

\begin{definition}[TS partition]\label{def:ts_partition}
A $(t,\varepsilon)$-\emph{Two-Separated partition} (TS partition) of a graph $G=(V,E)$ is a pair $(\cC,X)$ where  $X \subseteq V$, and $\cC$ is a partition of $V$, such that:
\begin{itemize}
    \item \emph{Two separation.} $\forall C,C'\in \cC$, every path $P=(v_1,\dots,v_\ell)$ from a vertex in $C\setminus X$ to a vertex in $C'\setminus X$ has two consecutive vertices in $X$. That is, for some $i \in [\ell-1]$, $v_i,v_{i+1} \in X$. 
    \item \emph{Bounded diameter.} $\forall C\in\cC$, $\weakdiam(C) \leq t$. We say that a TS partition is \emph{connected} if every $C \in \cC$ is connected.
    \item The \emph{cost ratio} of $(\cC,X)$, defined as
$R(\cC,X) = \max_i |C_i \cap X|/|C_i|$, holds $R(\cC,X)\le\varepsilon$.
\end{itemize}

Similarly, define for a TS partition $(\cC,X)$ the cost ratio of a cluster $C \in \cC$ as $|C \cap X|/|C|$. 
\end{definition}

Next, we introduce the notion of a decomposition into clusters with $\varepsilon$-cluster-degeneracy, and prove that when combined with low-diameter constraints, the existence of such a decomposition guarantees the existence of a $(t,\varepsilon)$-TS partition. This notion provides a convenient approach for constructing TS partitions in graphs, which we use in both \Cref{sec:minor_free_graphs} and \Cref{sec:general_graphs}.

Given a ordered partition  $\cC = \{C_1,\dots,C_k\}$, a node $v\in C_i$ is called a \emph{close-to-boundary} node if there is a node $u\in \cup_{j>i}C_j$ in a  cluster with a higher index than $i$, at distance at most $2$. We say that $\cC$ has $\varepsilon$-cluster-degeneracy if in every cluster $C_i$, at most $\eps$-fraction of its nodes are close-to-boundary nodes.

\begin{definition}
    A partition of $V$ into clusters $\cC = \{C_1,\dots,C_k\}$ has $\varepsilon$-cluster-degeneracy if  
\[\forall i,\quad
\left|\left\{ v\in C_{i}\mid \dist_{G}\left(v,V\setminus\cup_{j=1}^{i}C_{j}\right)\le2\right\} \right|\le\varepsilon\cdot|C_{i}|.
\]
\end{definition}

\begin{restatable}{lemma}{dsimpliests}
\label{lem:ds_implies_ts}
    Assume a graph $G = (V,E)$ has a partition into clusters $\cC = \{C_1,\dots,C_k\}$ that has $\varepsilon$-cluster-degeneracy and that is weakly $t$-bounded. Then $G$ has a $(t,\varepsilon)$-TS partition. Moreover, if all clusters in $\cC$ are connected, then $G$ has a connected $(t,\varepsilon)$-TS partition. 
\end{restatable}
    \begin{proof}
    Let $V_{> i} = V \setminus\cup_{j=1}^{i}C_j$, and $V_{> k} = \emptyset$. Setting $X = \cup_{i=1}^{k} \left(C_i \cap B_2(V_{> i})\right)$  we argue that $(\cC,X)$ is a $(t,\varepsilon)$-TS partition, i.e. it fulfills the following three properties:
    \begin{itemize}
        \item \emph{Diameter and connectivity:}  By assumption $\weakdiam(C) \leq t$ for all $C \in \cC$. Additionally, the partition is connected if and only if each $C \in \cC$ is connected.
        \item \emph{Cost ratio:} We note that $C_i \cap X = C_i \cap B_2(V_{> i})$, since $C_i \cap B_2(V_{> i}) \subseteq C_i$, and the clusters are disjoint. Therefore, by assumption we have $ |C_i \cap X|/|C_i| = |C_i \cap B_2(V_{> i})|/|C_i| \leq \varepsilon$.
        
        \item \emph{Two Separation:} Assume for the sake of contradiction that the two separation property fails. Let $P_{u,v}$ be the shortest path between two vertices $u,v\notin X$ that belong to different clusters $\cC$, such that $P_{u,v}$ does not have two consecutive vertices from $X$. 
        
        Denote $P_{u,v}=(w_0=u,w_1,\dots,w_{\ell-1},w_\ell=v)$. Let $u \in C_i, v \in C_j$ where we assume w.l.o.g. that $i < j$. Since $u \notin X$, we have $\dist(u,V_{>i})>2$, and specifically we have $\ell \geq 3$ and $w_1,w_2 \notin C_j$.

        If $w_1 \notin X$, then $P_{w_1,v}=(w_1,\dots,w_\ell=v)$  contradicts the minimality of $P_{u,v}$, thus $w_1\in X$. Similarly, we deduce that $w_2\in X$. Since $w_1,w_2 \in X$ are consecutive on $P_{u,v}$, we obtain a contradiction, and the claim follows.
    \end{itemize}
\end{proof}

In the remainder of the section, we show that the existence of $(t,\varepsilon)$-TS partitions implies a tradeoff between $O(t)$-PLS and $1$-PLS. To set this up, we first define the verification task of outputting a $(t,\varepsilon)$-TS partition, in which we both verify the existence of a $(t,\varepsilon)$-TS partition, and have each node locally output nearby clusters and their nodes in $X$.

\begin{definition}
\label{def:ts_partition_labeling}
    We say that a $t$-PLS $\Pi = (\Prov,\Ver)$ certifies and outputs a $(t,\varepsilon)$-TS partition if $\Pi$ is a $t$-PLS for the existence of a $(t,\varepsilon)$-TS partition in the graph according to \Cref{def:labeling_scheme}, and require an additional guarantee: assuming no node rejects for some labeling $\ell$, $\Ver$ outputs for each node $v$ a set of clusters $\Comp(v)$ and a set $X_C(v) \subseteq V$ for each $C \in \Comp(v)$, such that there exists a $(t,\varepsilon)$-TS partition $(\cC,X)$ of $G$, satisfying the following properties:
    \begin{enumerate}[(1)]
        \item \emph{Local cluster knowledge.} For any node $v \in V$, let $C_v \in \cC$ denote the unique cluster containing $v$. Then we have for any $v \in V$ that $\Comp(v) = \{C_u \mid u \in B_2(C_v)\}$. In other words, each node $v$ knows $C_v$ and any cluster that is of distance at most two from $C_v$.
        \item \emph{Local boundary knowledge.} For any $v \in V$, we have $X_C(v) = X \cap C$ for any $C \in \Comp(v)$.
    \end{enumerate}
    In this case, we refer to $(\cC,X)$ as the TS partition outputted by the labeling $\ell$.
\end{definition}

Finally, we formalize the relationship between the existence of TS partitions and a tradeoff theorem. Specifically, we show that a $(t,\varepsilon)$-partition that can be outputted with cost $c$ implies that any $1$-PLS of cost $p$ can be transformed into an $O(t)$-PLS with cost $O(\lceil \varepsilon \cdot p \rceil + c)$.

\begin{restatable}[TS partition implies tradeoff]{lemma}{tstradeoff}
\label{lem:tstradeoff}
Let $(\mathcal{G},\mathcal{I})$ be a configuration family. Suppose that for some functions $t =t(n) \ge 1$ and $\varepsilon = \varepsilon(n) \in [0,1)$, every graph $G \in \mathcal{G}$ with $|V(G)| = n$ admits a $(t,\varepsilon)$-TS partition, and that such a partition can be certified and outputted by an $O(t)$-PLS with cost $c = c(n)$. Then, for any predicate $P$ on $(\mathcal{G},\mathcal{I})$ and any $1$-PLS for $P$ with cost $p$, there exists an $O(t)$-PLS for $P$ with cost $O(\lceil \varepsilon\cdot p \rceil + c)$.
\end{restatable}

Before presenting the proof, we analyze basic properties of TS partitions, and simple lemmas related to certification of $P$. Then, we discuss efficient encoding and retrieval of information from clusters, which is a crucial part of the prover and verifier algorithms. 

\paragraph*{Basic properties and labeling of $(t,\varepsilon)$-TS partitions }
We start with a basic claim, that for every node $v$, there is at most a single cluster $C$ such that $B_1(v)$ ($v$ and its neighborhood) contains non-boundary vertices from $C$: $C\setminus X$.

\begin{lemma}
\label{lem:TS_unique_C_neighbor}
    Let $(\cC,X)$ be a $(t,\varepsilon)$-TS partition. Then for every $v \in V$, there exists a cluster $C \in \cC$ such that $B_1(v) \subseteq C \cup X$. In particular, $C$ is unique unless $B_1(v) \subseteq X$. 
\end{lemma}
\begin{proof}
    Suppose for the sake of contradiction that there is a vertex $v$ such that $B_1(v)$ is not contained in $C\cup X$ for any $C\in \cC$ (in particular $B_1(v)\nsubseteq X$).
    It follows that there are two neighbors $u_1,u_2$ of $v$, such that $u_1\in C_1\setminus X$ and $u_2\in C_2\setminus X$ for some clusters $C_1\ne C_2\in\cC$.
    But then $(u_1,v,u_2)$ is a $2$-path connecting $C_1 \setminus X$ and $C_2\setminus X$, contradiction to the two separation property. Uniqueness follows from the fact that clusters are disjoint, hence if $B_1(v) \subseteq C_1 \cup X$ and $B_1(v) \subseteq C_2 \cup X$ for some distinct clusters $C_1 \neq C_2$, then it follows that $B_1(v) \subseteq X$.
\end{proof}

Define $H_C = \{v \in V \mid B_1(v) \cap (C\setminus X) \neq \emptyset\}$. We have the following two properties of $H_C$.

\begin{lemma}
\label{lem:HC_complete}
    The following two properties hold.
    \begin{enumerate}[(1)]
        \item $\bigcup_{C \in \cC} H_C$ is equal to the set $\{u \mid B_1(u) \nsubseteq X\}$.
        \item $H_C \cap H_{C'} = \emptyset$ for any two distinct clusters $C,C' \in \cC$.
    \end{enumerate}
\end{lemma}
\begin{proof} We prove the two properties above.
\begin{enumerate}[(1)]
    \item This follows trivially from the definition of $H_C = \{v \in V \mid B_1(v) \cap (C\setminus X) \neq \emptyset\}$.
    \item Let $u \in H_C$. Then by \Cref{lem:TS_unique_C_neighbor} we have $B_1(u) \cap (C' \setminus X) = \emptyset$, or in other words $u \notin H_{C'}$.   
\end{enumerate}
\end{proof}

Let $\Pi_1 = (\Prov_1,\Ver_1)$ be a $1$-PLS for a property $P$. We introduce two notions for partial labelings related to $\Pi_1$ which help us certify that $(G,I)\in P$ using an $O(t)$-PLS. Namely, we define some notion of a good labeling of nodes of $X$ (according to $\Pi_1$), and for every $C$, given a good labeling of nodes of $X$, we define the notion of a good extension of the good labeling into $C \setminus X$. We then show that if there exists a good labeling of $X$, and a good extension of the labeling into $C \setminus X$ for every $C \in \cC$, then $(G,I)\in P$.


\begin{definition}
    Let $\Pi_1 = (\Prov_1,\Ver_1)$ be a $1$-PLS for a property $P$. Let $(\cC,X)$ be a $(t,\varepsilon)$-TS partition. We define the following.
    \begin{itemize}
        \item We say that a labeling $\ell_X : X \rightarrow \{0,1\}^{\leq p}$ is a good $X$-labeling if for all nodes $v \in X$ such that $B_1(v) \subseteq X$, we have 
    $\Ver_1(G,I,\ell_{X},v) = \accept.$
    \item Let $\ell_X$ be a good $X$-labeling. For a cluster $C \in \cC$, a function $\ell_{C \cup X}: (C \cup X) \rightarrow \{0,1\}^{\leq p}$ is called a good extension of $\ell_X$ if $\ell_{C \cup X}(u) = \ell_X(u)$ for every $u \in X$, and for every $v \in H_C$ we have
    $\Ver_1(G,I,\ell_{C \cup X},v) = \accept.$
    \end{itemize}
\end{definition}

\begin{lemma}\label{lem:tstradeoff_good_extensions}
    Let $\Pi_1 = (\Prov_1,\Ver_1)$ be a $1$-PLS for a property $P$. Let $(\cC,X)$ be a $(t,\varepsilon)$-TS partition, and let $\ell_X : X \rightarrow \{0,1\}^{\leq p}$ be a good $X$-labeling. If for every cluster $C \in \cC$ there exists a good extension $\ell_{X \cup C}$, then $(G,I) \in P$.
\end{lemma}
\begin{proof}

    Let $\ell_X$ be a good $X$-labeling, and let $\ell_{C_1 \cup X},\dots,\ell_{C_k \cup X}$ be good extensions of $\ell_X$ for clusters $C_1,\dots,C_k$, which exist by the lemma's assumption. We define a labeling  $\ell_V : V \rightarrow \{0,1\}^{\leq p}$ as follows.
    \[\ell_V(u) = \begin{cases}
             \ell_{X}(u), & u \in X, \\
            \ell_{C \cup X}(u),  & u \in C \setminus X.
        \end{cases}\]
    We argue that $\Ver_1(G,I,\ell_V,u) = \accept$ for all $u \in V$, which implies that $(G,I) \in P$ by the soundness of $(\Prov_1,\Ver_1)$. We split into two cases - when $B_1(u) \subseteq X$ and otherwise. In the first case, we have by definition of $\ell_X$ being good that 
    $\Ver_1(G,I,\ell_V,u) = \Ver_1(G,I,\ell_X,u) = \accept$. In the other case, we have by \Cref{lem:HC_complete} that $u \in H_C$ for some unique $C \in \cC$. Therefore, by definition of $\ell_{C \cup X}$ being a good extension of $\ell_X$, we get $\Ver_1(G,I,\ell_V,u) = \Ver_1(G,I,\ell_{X \cup C},u) = \accept$.
\end{proof}

Finally, we discuss how to encode a string across a cluster using node labels, and how to decode it again in order to retrieve the stored information.

    \paragraph*{Encoding and decoding strings on clusters.} As motivated in the introduction, we would like to be able to encode a good $\ell_X$ on the labels of the network, where $\ell_X(v)$ for $v \in C \cap X$ is encoded using labels of cluster $C$'s vertices. We first discuss encoding and decoding strings across a cluster, which uses the internal identifier ordering to describe the part numbers of string segments spread across the cluster.

Given a cluster $C \subseteq V$ and a string $S$, we define the lexicographic encoding $\lex$ of $S$ on $C$ as follows. Let $k = |C|$, and let $u_1,\dots,u_k$ be the nodes of $C$, ordered lexicographically by their node identifiers. Let $S' = \Pad(S,k)$ be a padding of $S$ so that $|S'|$ is divisible by $k$, and let $S'_1,\dots,S'_k$ denote the strings obtained by partitioning $S'$ into $k$ substrings, each of the same length. For node $u_i$, we define $\lex(C,S,u_i) = S'_i$.\footnote{We refer back to the string notations in \Cref{sec:prelims} for the introduction of $\Pad$ and the fixed delimiter scheme.}

Relatedly, we also define a decoding function $\decodelex$. Given a cluster $C$ and a labeling $\ell$ on its vertices such that $\ell(u_i) = \lex(C,S,u_i)$ for some string $S$, the function $\decodelex(C,\ell)$ reconstructs $S$ by concatenating the labels in lexicographic order and then removing the padding.

We are now ready to prove \Cref{lem:tstradeoff}.

\begin{proof}[Proof of \Cref{lem:tstradeoff}]
    Given $(\Prov_{\TS},\Ver_{\TS})$, an $O(t)$-PLS for computing a $(t,\varepsilon)$-TS partition, and $(\Prov_1,\Ver_1)$ a $1$-PLS for certifying $P$, we construct $(\Prov_t,\Ver_t)$ as follows. We assume w.l.o.g. that the labels given by $\Prov_t$ and received by $\Ver_t$ for each node are of the form $(\var{TSLabel},\var{ProofPart})$ where $\var{TSLabel},\var{ProofPart}$ are some strings. If a label breaks this format, $\Ver_t$ rejects in the corresponding node.

    \paragraph*{Description of $\Prov_t$:}
    Let $(\cC,X)$ be a $(t,\varepsilon)$-TS partition of the graph that is outputted by $\Prov_{\TS}$ on the graph $G$. 
    For a node $v$, denote $C_v$ as the unique cluster containing $v$. 
    Let $u_1,\dots,u_k$ be the nodes of $C_v \cap X$, ordered lexicographically w.r.t. their node identifier, and let 
    \[\var{XLabel}(C_v,X) = (\Prov_1(G,I,u_1),\dots,\Prov_1(G,I,u_k)).\] 
    
    The mapping $\Prov_t$ gives each node $v$  as its label the tuple 
    \[\Prov_t(G,I,v) = \left(\Prov_{\TS}(G,I,v)~,~ \lex(C_v,\var{XLabel}(C_v,X),v)\right).\] 
    By the lemma's assumption, $\Prov_{\TS}(G,I,v)$ is of size at most $c$. For any $C \in \cC$ we have $|\var{XLabel}(C,X)| \leq p \cdot |C \cap X|$. Therefore, 
    \[|\lex(C_v,\var{XLabel}(C_v,X),v)| = O(\lceil p \cdot |C_v \cap X|/|C_v| \rceil) = O(\lceil p\cdot \varepsilon \cdot |C_v|/|C_v| \rceil) = O(\lceil \varepsilon \cdot p \rceil).\] 
    Hence the total cost of $\Prov_t$ is $O(\lceil \varepsilon \cdot p \rceil+c)$ as required.

    \paragraph*{Description of $\Ver_t$:}     
    Recall we assume that the proof labels of every node $v$ are given in a format encoding two arbitrarily large strings $(\var{TSLabel}(v),\var{ProofPart}(v))$ (otherwise we have that node reject). Let $\var{TSLabel} : V \rightarrow \{0,1\}^*, \var{ProofPart}: V\rightarrow \{0,1\}^*$ be the corresponding global functions/labelings. We describe $\Ver_t$ as a distributed algorithm, where we say that a node accepts or rejects corresponding to whether $\Ver_t$ accepts or rejects on this node.

    In the first step of the verifier we have every node $v$ run $\Ver_{\TS}$ on the labeling $\var{TSLabel}$ to certify and output some $(t,\varepsilon)$-TS partition. By the soundness of $\Pi_{\TS}$, either a node rejects, or there exists some $(t,\varepsilon)$-TS partition $(\cC,X)$ satisfying the properties of \Cref{def:ts_partition_labeling}. Assuming no node rejects, we fix this partition $(\cC,X)$. Let $C_v$ denote the unique cluster containing $v$, and denote for each cluster $C \in \cC$ by $\leader(C)$ as the node with the maximum node identifier in $C \setminus X$.

    From the labeling $\var{ProofPart}$ we globally define a labeling function $\ell_X : X \to \{0,1\}^*$ for all nodes in $X$, in the following manner. Let $C \in \cC$ be a cluster, and let $C \cap X = \{u_1,\dots,u_k\}$ be the nodes of $C \cap X$ ordered lexicographically by node id. Then we define,

\[(\ell_X(u_1),\dots,\ell_X(u_{k})) = \decodelex(C,\var{ProofPart}).\]

Node $v = \leader(C)$ asserts that $\decodelex(C,\var{ProofPart})$ does indeed encode $k = |C \cap X|$ strings, rejecting otherwise. 

We note that to compute $\ell_X(v)$ of any node $v \in C \cap X$ it is sufficient to know the topology of $C$ with its identifiers, and its labels of $\var{ProofPart}$. Therefore, any node $v$ can locally reconstruct for each of $C \in \Comp(v)$ the labels $\ell_X(u)$ of all nodes $u \in C \cap X$, and specifically for $u \in B_2(C_v) \cap X$.

Recall the definition of $H_C = \{v \mid B_1(v) \cap (C \setminus X) \neq \emptyset\}$. Every node $v$ performs the following. 
\begin{enumerate}[(1)]
    \item If $B_1(v) \subseteq X$ (including $v$), node $v$ asserts that $\Ver_1(G,I,\ell_X,v) = \accept.$
    \item If $v = \leader(C)$ for some cluster $C \in \cC$, node $v$ asserts that there exists a labeling $\ell_{C \setminus X}:C \setminus X \rightarrow \{0,1\}^{\leq p}$ such that the extension $\ell_{C \cup X} = \ell_X \cup \ell_{C \setminus X}$ is good, or in other words that
    $\forall_{u \in H_C} \Ver_1(G,I,\ell_{C \cup X},u)= \accept$. We note that $\leader(C)$ can perform this since in order to check that the extension is good, $\leader(C)$ only needs to know $\ell_X(u)$ for $u \in B_2(C) \cap X$.
    \item If all assertions passed, node $v$ accepts.
\end{enumerate}

Each node checks the assertions above, and accepts if and only if no assertion fails. This concludes the description of $\Ver_t$. We summarize all assertions made by the verifier in the following list, from the viewpoint of a node $v$ (skipping some in-between steps).
\begin{enumerate}[(1)]
    \item \label{item:ts_assert_2} Node $v$ runs $\Ver_{\TS}$ on the labeling $\var{TSLabel}$, and rejects if it rejects. Otherwise, it obtains its relevant part in $(\cC,X)$, i.e. $\Comp(v) \subseteq \cC$ and for each $C \in \Comp(v)$ a set $X_C(v)$.
    \item \label{item:ts_assert_3} If $v = \leader(C)$ for some $C \in \Comp(v)$, it asserts whether $\decodelex(C,\var{ProofPart})$ encodes $|C \cap X|$ many strings (which correspond to $(\ell_X(u))_{u \in C \cap X}$).
    \item \label{item:ts_assert_4} If $B_1(v) \subseteq X$ (including $v$), node $v$ asserts that $\Ver_1(G,I,\ell_X,v) = \accept.$
    \item \label{item:ts_assert_5} If $v = \leader(C)$ for some $C \in \Comp(v)$, then it asserts that there exists a labeling $\ell_{C \setminus X}$ of $C \setminus X$  such that the extension $\ell_{C \cup X} = \ell_X \cup \ell_{C \setminus X}$ is good, i.e. that
    $\forall_{u \in H_C} \Ver_1(G,I,\ell_{C \cup X},u) = \accept$.
\end{enumerate}

        \paragraph*{Correctness:} To show correctness, we assume that a $(t,\varepsilon)$-TS partition exists, and show that  $\Ver_t(G,I,\Prov_t(G,I),v) = \accept$, i.e. that all nodes accept given the honest prover labeling. 

        We have that $\var{TSLabel}(v) = \Prov_{\TS}(G,I,v)$, therefore by correctness of $(\Prov_{\TS},\Ver_{\TS})$ no node rejects, and we have that each node obtains a set $\Comp(v) \subseteq \cC$ and $X_C$ for each $\Comp(v)$ satisfying the guarantees of \Cref{def:ts_partition_labeling}, passing Assertion~(\ref{item:ts_assert_2}).

        Recall the definition of 
        $\var{XLabel}(C,X) = (\Prov_1(G,I,u))_{u \in C \cap X}$. For each node $v$, we have that $\var{ProofPart}(v) = \lex(C,\var{XLabel}(C_v,X),v)$, therefore for each $C \in \Comp(v)$ we have that \[\decodelex(C,\var{ProofPart}) = (\Prov_1(G,I,u))_{u \in C \cap X}.\] In particular, we do not reject in Assertion~(\ref{item:ts_assert_3}), and for every node $v \in V$, we have that it can compute $\ell_X(u)$ for all $u \in C \cap X$.

        Finally, we show that a node $v$ does not reject in Assertion~(\ref{item:ts_assert_4}),(\ref{item:ts_assert_5}). For Assertion~(\ref{item:ts_assert_4}), we note that all of $u \in B_1(v)$ have $\ell_X(u) = \Prov_1(G,I,u)$, and therefore $v$ does not reject, by correctness of $\Prov_1$. For Assertion~(\ref{item:ts_assert_5}), we note that by correctness of $\Prov_1$, we have that $\ell_X$ is a good labeling, and hence the labeling $\ell_{C \cup X}(u) = \Prov_1(G,I,u)$ is a good extension to $\ell_X$, and $v = \leader(C)$ accepts. 

        \paragraph*{Soundness:} 
        Assume that all nodes accept. By Assertion~(\ref{item:ts_assert_2}),(\ref{item:ts_assert_3}) we have some TS partition $(\cC,X)$ which the nodes collectively hold, and that each cluster $C \in \cC$ encodes $|C \cap X|$ many strings, corresponding to some labeling $(\ell_X(u))_{u \in C \cap X}$. Therefore, the labeling $\ell_X$ is well-defined. By Assertion~(\ref{item:ts_assert_4}), we have that $\ell_X$ is a good $X$-labeling according to \Cref{def:ts_partition_labeling}, and moreover by Assertion~(\ref{item:ts_assert_5}), we have for each $C \in \cC$ a good extension $\ell_{C \cup X}$. Therefore, by \Cref{lem:tstradeoff_good_extensions}, we have that $(G,I) \in P$, and the claim follows.

\end{proof}

\section{Tradeoff Theorem for Minor-Free Graphs}
\label{sec:minor_free_graphs}
In this section, our goal is to prove that the Tradeoff Conjecture holds for $K_r$-minor free graphs $\cG$ for any constant $r \geq 1$, up to a single additive $O(\log{n})$ factor. This trivially implies the conjecture for any $H$-minor free graph family, since it is $K_{r}$-minor free as well, for $r = |V(H)|$.

\thmminorfree*

On a high level, we show \Cref{thm:main_minor_free} using the existence of a good padded decomposition (introduced next) in $\cG$, and relating its existence to the existence of a good TS partition. Then we use the fact that $O(\log{n})$ is always sufficient for an $O(t)$-PLS to certify and output a TS partition.

\begin{definition}[Padded Decomposition]\label{def:PadDecompostion}
    
    Let $G = (V,E)$ be a graph. A distribution $\mathcal{D}$ over partitions of $V$ is a $(\beta,\delta,\Lambda)$-padded decomposition if every $\mathcal{P}\in{\rm supp}(\mathcal{D})$ is weakly $\Lambda$-bounded, and for every $0\le\gamma\le\delta$ and $z\in V$, the ball $B_{\gamma\Lambda}(z)$ satisfies
    $\Pr[B_{\gamma\Lambda}(z) \subseteq C_z] \ge e^{-\beta\gamma}$ ($C_z$ denotes the unique cluster in $\cP$ containing $z$).    
    We say that $G$ admits a $(\beta,\delta)$-padded decomposition scheme if for every $\Lambda>0$, there is a $(\beta,\delta,\Lambda)$-padded decomposition for $G$.
\end{definition}

We use a result by Conroy and Filtser \cite{CF25}, that shows that $K_r$-minor free graphs have a good padded decomposition.

\begin{theorem}[\cite{CF25}]
\label{thm:padded_k_r}
    For any $r \geq 1$, every $K_r$-minor free graph admits an $(O(\log{r}),O(1))$-padded decomposition scheme.
\end{theorem}

Next, we show that having a good padded decomposition implies a good TS partition. First, we show that given an ``active'' set of nodes $L$, we can construct a cluster $C \subseteq L$ such that all vertices, other than a $O(\beta/t)$ fraction, are at least of distance $2$ from all other nodes in $L$. We then use this iteratively to get a partition with $O(\beta/t)$-cluster-degeneracy, and conclude the existence of a $(t,O(\beta/t))$-TS partition.

\begin{lemma}
\label{lem:minor_free_ts_existance}
    Let $G=(V,E)$ be a graph that admits a $(\beta,O(1))$-padded decomposition scheme. Let $L \subseteq V$ be a set of nodes. Then for any $t=\Omega(1)$, there is a cluster $C\subseteq L$ that has weak diameter $\le  t$, and such that $|C \cap B_2(L \setminus C)| \le\frac{2\beta}{t}\cdot|C|$.
\end{lemma}
\begin{proof}
\sloppy
    We sample a partition $\cP = \{C_1,\dots,C_k\}$ from a $(\beta,O(1),t)$-padded decomposition. Denote $C_v$ the unique cluster containing $v$. Let $\Gamma_2(L) =\left\{v\in L\mid B_2(v)\nsubseteq C_v\right\}$. We note that $\Gamma_2(L)$ is precisely the set of second-layer boundary vertices in $L$ across all clusters of $\cP$.
    For every vertex $v$, it holds that
    \[
    \Pr\left[v\in\Gamma_2(L)\right] \leq \Pr\left[B_2(v)\nsubseteq C_v\right]\le1-e^{-\frac{2\beta}{t}}\le\frac{2\beta}{t}~.
    \]
    Thus by linearity of expectation,  $\mathbb{E}[|\Gamma_2(L)|]\le \frac{2\beta}{t}\cdot |L|$.
    In particular, there is some partition $\cP$ where $|\Gamma_2(L)|\le \frac{2\beta}{t}\cdot |L|$.
    Let $\cP$ be such a partition, and let $\cC$ be the resulting clusters, and let $\cC' = \{C \cap L \mid C \in \cC\}$. Let $C'_1,\dots,C'_k$ denote the non-empty clusters of $\cC'$. Then the following holds,
 \[
\min_{i \in [k]}\frac{\left|C'_i \cap B_2(L \setminus C'_i) \right|}{\left|C'_{i}\right|}\le\frac{\sum_{i=1}^{k}\left|C'_i \cap B_2(L \setminus C'_i)\right|}{\sum_{i=1}^{k}\left|C'_{i}\right|} \leq \frac{|\Gamma_2(L)|}{|L|} \leq \frac{2\beta}{t}~.
\]
Where the first inequality follows from the mediant inequality.\footnote{Recall that the mediant inequality states that if $a_1,b_1,\dots,a_\ell,b_\ell >0$ are real numbers, then $\min_{i \in [\ell]}\frac{a_i}{b_i} \leq \frac{\sum_{i=1}^\ell a_i}{\sum_{i=1}^\ell b_i}$.} Thus taking $C = C'_i$ as the cluster minimizing the expression on the left-hand side, we found a cluster $C$ as desired.
\end{proof}

\begin{lemma}
\label{lem:padded_ts}
    Let $G = (V,E)$ be a graph that admits a $(\beta,O(1))$-padded decomposition scheme. Then $G$ has a $(t,2\beta/t)$-TS partition.
\end{lemma}
\begin{proof}
     Consider the following iterative procedure, where we maintain an active set $\alive$, and construct in each step a cluster $C \subseteq \alive$ using \Cref{lem:minor_free_ts_existance}, and remove $C$ from $\alive$ until the set is empty. 
    
    Initially, set $\alive_1 = V$. In step $i \geq 1$, assume we have a set $\alive_i \subseteq V$, satisfying $|\alive_i| > 0$. By \Cref{lem:minor_free_ts_existance}, there exists a cluster $C_i$ of weak diameter $\le  t$, such that $|C_i \cap B_2(\alive_i \setminus C_i)| \le 2\beta|C_i|/t$. Define $\alive_{i+1} = \alive_i \setminus C_i$. We repeat the process until $\alive_{i+1}$ is empty. 
    
    By a simple induction, we have $\alive_i = V \setminus \cup_{j = 1}^{i-1} C_j$ for all $i \in [k]$. Therefore, the set of obtained clusters $\cC = \{C_1,\dots,C_k\}$ are a partition of $V$, having each weak diameter $\leq t$, and satisfy 
    \[\forall_i \frac{|C_i \cap B_2(V \setminus \cup_{j=1}^{i}C_j)|}{|C_i|} = \frac{|C_i \cap B_2(\alive_i \setminus C_i)|}{|C_i|} \le 2\beta/t.\] Therefore, $\cC$ has $(2\beta/t)$-cluster-degeneracy, and by \Cref{lem:ds_implies_ts}, $G$ has a $(t,2\beta/t)$-TS partition.
\end{proof}

By combining \Cref{thm:main_minor_free} and \Cref{lem:padded_ts}, we get the following as corollary.

\begin{corollary}
\label{cor:minor_free_ts}
    For any $r \geq 1$, the family of $K_r$-minor free graphs has a $(t,\Theta(\log{r}/t))$-TS partition.
\end{corollary}
 
Next, we show that a $(t,\varepsilon)$-TS partition can always be certified and outputted using an $O(t)$-PLS with cost $O(\log{n})$. We remark that this lemma is fairly standard, and similar in spirit to prior works (e.g. \cite{FOS21}).

\begin{restatable}{lemma}{tsminorcompute}
\label{lem:tsminorcompute}
    In any graph family $\cG$ and any $t(n) \geq 1,\varepsilon(n) \in [0,1)$ there exists a $O(t)$-PLS that certifies and outputs a $(t(n),\varepsilon(n))$-TS partition with cost $O(\log{n})$.
\end{restatable}
    
The proof of \Cref{lem:tsminorcompute} is deferred to \Cref{sec:app_minor}. We are now ready to prove the main theorem.

\begin{proof}[Proof of \Cref{thm:main_minor_free}]
    Recall that by \Cref{lem:tstradeoff}, it suffices to show that $\cG$ has a $(t,\Theta(1/t))$-TS partition, and that it can be certified and outputted by an $O(t)$-PLS with cost $O(\log{n})$. The theorem follows immediately by combining \Cref{cor:minor_free_ts} and \Cref{lem:tsminorcompute}.
\end{proof}

As an immediate corollary, we can obtain a weaker version of \Cref{thm:general_main} for general graphs, which costs an additional logarithmic additive factor, using the following theorem of Bartal.

\begin{theorem}[\cite{Bar96}]
\label{thm:bartal}
    Every $n$-vertex graph admits an  $(O(\log{n}),O(1))$-padded decomposition scheme.
\end{theorem}

Again, by \Cref{lem:tstradeoff} combined with \Cref{lem:padded_ts}, \Cref{lem:tsminorcompute} and \Cref{thm:bartal}, we obtain the following.

\begin{corollary}
        Let $(\cG,\cI)$ be a configuration family, and $P$ a predicate on $(\cG,\cI)$. For $G \in \cG$ with $|V(G)| = n$, if there exists a $1$-PLS for $P$ with cost $p$, then for any $t \geq 1$ there exists an $O(t\log{n})$-PLS for $P$ with cost $O(\lceil p/t \rceil+\log{n})$.
\end{corollary}

\section{Tradeoff Theorem for General Graphs}
\label{sec:general_graphs}

Our goal in this section is to show a near resolution of the Tradeoff Conjecture in general graphs, up to a single multiplicative logarithmic factor, and with only constant additive factors. In particular, the cost of this scheme converges to a constant number of bits as $t$ grows.  

\thmgeneral*

\subsection{Helpful Sequential Algorithm for Finding a \texorpdfstring{$(O(t\log{n}),1/t)$}{(O(t log n), 1/t)}-TS Partition}
Before proving \Cref{thm:general_main}, we describe and analyze a sequential randomized algorithm $\cA$ on a graph $G$, which obtains a $(O(t\log{n}),1/t)$-TS partition with high probability, and is later used as part of the labeling scheme. For ease of notation, we treat $\log{n}$ as an integer throughout the section, and note that this assumption can be easily removed using appropriate rounding functions.

\paragraph*{Description of algorithm $\cA$:}
Let $v_1,\dots,v_n$ be an ordering of $V$ according to lexicographic order. Initially, we initialize our set of clusters as $\cC = \emptyset$, and the set of active nodes as $\alive_1 = V$. We iterate over the nodes according to the ordering. In step $i \geq 1$, we check whether $\alive_i \cap B_{2t\log{n}}(v_i) = \emptyset$, and if so, we continue to step $i+1$. Otherwise, we choose a random even integer $r_i \in \{2t\log{n}+2,2t\log{n}+4,\dots,8t\log{n}\}$, and define $C_i = B_{r_i}(v_i)\cap \alive_i$. We note that $C_i$ does not necessarily contain $v_i$, and also might be disconnected. We add $C_i$ to $\cC$ if 
\begin{equation}\label{eq:general_expands}
   \frac{|B_{r_i}(v_i)\cap \alive_i|}{|B_{r_i-2}(v_i)\cap \alive_i|} \leq 1+\frac{1}{t}~.
\end{equation}
If $C_i$ was added to $\cC$, we set $\alive_{i+1} = \alive_{i} \setminus C_i$, set $X_i = C_i \cap B_2(\alive_i \setminus C_i)$. Otherwise, we set $X_i = \emptyset$ and $\alive_{i+1} = \alive_{i}$. Following this, we continue to step $i+1$. We conclude at the end of step $n$, i.e. after the final vertex $v_n$ is processed. Let $X = \bigcup_{i=1}^n X_i$. We output $(\cC,X)$ as our TS partition. 

\paragraph*{Analysis of algorithm $\cA$:} In the analysis, our goal is to show that with some high probability the resulting $(\cC,X)$ is an $(O(t\log{n}),1/t)$-TS partition of $G$. Specifically, we need to show that $\cC$ is a partition of $V$ (i.e., every node is assigned to exactly one cluster in $\cC$), that the clusters are weakly $t$-bounded, have cost ratio at most $1/t$, and are two separated. We first show that $\alive_{n+1} = \emptyset$, which implies that $\cC$ is a partition. Following this, we show that conditioned that $\cC$ is a partition, then $(\cC,X)$ is a $(O(t\log{n}),1/t)$-TS partition. 

We next show that in each step $i$, if we have some node in $\alive_i$ in the vicinity of $v_i$ (i.e. $\alive_i \cap B_{2t\log{n}}(v_i) \neq \emptyset$), then with some constant probability we create a new cluster in $\cC$ that contains these nodes. We later argue that since each node has $\Theta(t\log{n})$ nodes in its $\Theta(t\log{n})$-neighborhood (unless the graph is smaller than $\Theta(t\log{n})$), then with high probability each node is added to a cluster in some step, and the correctness of the algorithm follows.

\begin{lemma}
\label{lem:general_good_cluster}
    If at the start of step $i$ we have $\alive_i \cap B_{2t\log{n}}(v_i) \neq \emptyset$, then Eq.~(\ref{eq:general_expands}) holds with probability at least $1/2$ over the choice of $r_i \in [2t\log{n}+2,8t\log{n}]$. 
\end{lemma}
\begin{proof}
    Fix a step $i$, and denote 
    \[\Bad_i = \{r \in \{2t\log{n}+2,2t\log{n}+4,\dots,8t\log{n}\} \mid r \text{ does not satisfy } Eq.~(\ref{eq:general_expands}) \text{ in step $i$}\}.\] Assume by contradiction that $|\Bad_i| \geq t\log{n}+1$, and let  $\{r'_1,\dots,r'_{t\log{n}+1}\} \subseteq \Bad_i$, ordered in ascending order. Therefore, for all $j \in [2,t\log{n}+1]$ we have that
    \[|B_{r'_{j}}(v_{i})\cap\alive_{i}|\geq(1+\frac{1}{t})|B_{r'_{j}-2}(v_{i})\cap\alive_{i}|\geq(1+\frac{1}{t})|B_{r'_{j-1}}(v_{i})\cap\alive_{i}|,\]
    where the first inequality holds by definition of $\Bad_i$, and the second inequality holds due to the fact that we considered only even indices, hence $r'_{j-1} \leq r'_{j}-2$, and the fact that a ball can only increase in size as the radius grows. By induction on $j$, we get that
    \[|B_{r'_{t\log{n}+1}}(v_{i})\cap \alive_i| \geq (1+\frac{1}{t})^{t\log{n}}|B_{r'_1}(v_i)\cap \alive_i| \geq (1+\frac{1}{t})^{t\log{n}}|B_{2t\log{n}}(v_i)\cap \alive_i|  > n,\]
    where the last inequality follows by the fact $(1+1/t)^t \geq 2$ for any $t \geq 1$, and by the lemma's assumption that $|B_{2t\log{n}}(v_i)\cap \alive_i| \geq 1$.
\end{proof}

\begin{lemma}
\label{lem:general_alive_empty}
    Assume $n \geq 2t\log{n}$. Then $\Pr(\alive_{n+1} \neq \emptyset) \leq 1/n$. 
\end{lemma}
\begin{proof}
We consider a node $v$ and bound the probability of $v \in \alive_{n+1}$. Let $S = B_{2t\log{n}}(v)$, denote $k = |S|$, and let $v_{i_1},\dots,v_{i_k}$ be the vertices of $S$ in lexicographic order. Since we assume $n \geq 2t\log{n}$, it is easy to see that $k \geq 2t\log{n}$. When processing each $v_{i_j}$ for $j \in [k]$, if $v$ is active in this step (i.e. $v \in \alive_{i_j}$), then by \Cref{lem:general_good_cluster}, $v$ is still active at the end of the step (i.e. $v \in \alive_{i_j+1}$) with probability at most $1/2$. Therefore, 
    \[\Pr(v \in \alive_{n+1}) \leq \prod_{j=1}^k \Pr(v \in \alive_{i_j+1} \mid v \in \alive_{i_j}) \leq (1/2)^{k} \leq (1/2)^{2t\log{n}} \leq \frac{1}{n^2}.\]
    and $\Pr(\alive_{n+1} \neq \emptyset) \leq 1/n$ follows by the union bound on all vertices $v \in V$.
\end{proof}

Next, we show that if $\cC$ is a partition of $V$, then it is an $(O(t\log{n}),1/t)$-TS partition.

\begin{lemma}
\label{lem:general_if_partition_then_ts}
    If $\alive_{n+1} = \emptyset$, then $(\cC,X)$ is a $(16t\log{n},1/t)$-TS partition.
\end{lemma}
\begin{proof}
    Recall that if $C_i$ is added to $\cC$, then it is non-empty, and that $C_i = B_{r_i}(v_i) \cap \alive_i$. Since $r_i \leq 8t\log{n}$, it follows that $\weakdiam(C_i) \leq 16t\log{n}$. Next, we show that the cost ratio and two separation conditions hold. The cost ratio follows by: 

        \begin{flalign*}
|X_{i}| & =|C_{i}\cap B_{2}(\alive_{i}\setminus C_{i})|\\
 & \le|B_{r_{i}}(v_{i})\cap\alive_{i}|-|B_{r_{i}-2}(v_{i})\cap\alive_{i}|\\
 & \le(1+\frac{1}{t})\cdot|B_{r_{i}-2}(v_{i})\cap\alive_{i}|-|B_{r_{i}-2}(v_{i})\cap\alive_{i}|\\
 & =\frac{|B_{r_{i}-2}(v_{i})\cap\alive_{i}|}{t}\le\frac{|C_{i}|}{t}~.
\end{flalign*}
We note that indeed $X_i = C_i \cap X$, since $X_i \subseteq C_i$ and the clusters are disjoint. Finally, we show the two separation property.

Assume for the sake of contradiction that the two separation property fails. Let $P_{u,v}$ be the shortest path between two vertices $u,v\notin X$ that belong to different clusters in $\cC$ such that $P_{u,v}$ does not have two consecutive vertices from $X$. By minimality $\dist_G(u,v)\le 2$ (as otherwise we could have taken a shorter path). Let $C_i,C_j\in\cC$ such that $u\in C_i$, $v\in C_j$, and suppose w.l.o.g. that $i < j$. As $\dist_G(u,v)\le|P_{u,v}|\le 2$, it follows that $u\in C_i \cap B_2(\alive_j)$, and thus $u\in X$, a contradiction.

\end{proof}

Combining \Cref{lem:general_alive_empty} and \Cref{lem:general_if_partition_then_ts}, we immediately get the following as a corollary.

\begin{corollary}
\label{cor:general_alg_is_ts}
    Assume $n \geq 2t\log{n}$. Then $(\cC,X)$ is a $(16t\log{n},1/t)$-TS partition with probability at least $1-(1/n)$. 
\end{corollary}

This concludes the description and analysis of $\cA$. We call a randomness string $R$ \emph{good} w.r.t. graph $G$ if by running $\cA$ using randomness $R$, we have that the resulting clusters $(\cC,X)$ are a $(16t\log{n},1/t)$-TS partition. In particular, \Cref{cor:general_alg_is_ts} implies that for any graph $G$, at least a $1-(1/n)$ fraction of the randomness strings are good w.r.t. $G$. 

\paragraph*{Encoding the randomness string of $\cA$ and reducing its size:} Let us consider the randomness used by $\cA$, and its encoding. Naively, we can represent it as a choice of $n$ even indices $r_i \in [2t\log{n}+2,8t\log{n}]$ for each $i \in [n]$. However, in our labeling scheme we run a sort of simulation of $\cA$ in a distributed manner, where nodes have access to the shared randomness string. The naive representation does not lend itself well to that sort of distributed implementation, as it requires knowledge of the global ordering of node identifiers (i.e., knowing the lexicographic rank of each node $u$ among all identifiers). Instead, we view the randomness as a function that assigns each possible identifier a random index, namely $R : [\mathrm{poly}(n)] \rightarrow [2t\log{n}+2,8t\log{n}]$. With this representation, a node $u$ can locally compute its random index from its identifier alone. The drawback is that describing $R$ explicitly requires $\mathrm{poly}(n)$ random bits. We therefore apply a randomness reduction argument to reduce it to $O(\log{n})$ bits, which we show later can be shared among all nodes using only $O(1)$ label bits. 

Let $\cG_n$ be the family of all graphs of size $|V(G)| = n$. Next, we prove using a standard argument that for graphs in $\cG_n$ there exists a set of randomness strings $\{R_1,\dots,R_{n^2}\}$ such that for each $G \in \cG_n$ there exists an $R_i$ that is a good randomness string for $G$.\footnote{In the literature, this type of argument is often referred to as ``Newman's trick'' \cite{N91}. It was originally used in communication complexity to show that public coin protocols can be simulated by private coin protocols, and as a general method to reduce randomness.}

\begin{lemma}
\label{lem:general_good_randomnesss}
    For any integer $n$, there exists a set $F_n:[n^2] \rightarrow \{0,1\}^*$ such that for any graph $G \in \cG_n$ there exists an index $i \in [n^2]$ such that $F_n(i)$ is a good randomness string for $G$. 
\end{lemma}
\begin{proof}
 We notice that algorithm $\cA$ uses for graphs in $\cG_n$ at most $O(n\log{n})$ random bits. Choosing a randomness string $R$ uniformly at random from $\{0,1\}^{O(n\log{n})}$, we have that for a fixed $G \in \cG_n$ that
    \[\Pr(R \text{ is good randomness for } G) \geq 1-(1/n).\]
    Therefore, if we choose $n^2$ strings $R_1,\dots,R_{n^2} \in \{0,1\}^{O(n\log{n})}$ uniformly at random, we have that  
    \[\Pr(\exists_i R_i \text{ is good randomness for } G) \geq 1-(1/n)^{n^2}.\]
    By union bound over all graphs in $\cG_n$ we have that
    \[\Pr(\forall_{G \in \cG_n}\exists_i R_i \text{ is good randomness for } G) \geq 1-2^{\binom{n}{2}} \cdot (1/n)^{n^2} > 0.\]
    Therefore, there exists a set of $n^2$ randomness strings $R_1,\dots,R_{n^2}$ such that for all $G \in \cG_n$, there exists an index $i \in [n^2]$ such that $R_i$ is a good randomness string for $G$. Setting for $i \in [n^2]$ the function $F_n(i) = R_i$, the claim follows.
\end{proof}

\subsection{Labeling Scheme for General Graphs} Next, we show how to allow nodes to locally simulate $\cA$ given $O(1)$-sized labels. Before showing the labeling scheme, we show for any $r \geq 1$ an $O(r)$-PLS that allows the honest prover to share a string $S$ of size $r$ across all nodes, with cost $O(1)$. Informally, we use this scheme to inform all nodes of an index of a good randomness string for $G$.

\begin{lemma}
\label{lem:shared_string}
    Let $r \geq 1$ be an integer. Then for any $S \in \{0,1\}^{r}$, there exists an $O(r)$-PLS $(\Prov,\Ver)$ that uses $O(1)$ bits per node, such that for the honest prover, all nodes output $S$, and for any other labeling, either all nodes output some shared string $S'$, or at least one node rejects.
\end{lemma}

The proof of \Cref{lem:shared_string} is deferred to \Cref{sec:app_general}, and we give a brief sketch of proof. 
\begin{proof}[Proof sketch of \Cref{lem:shared_string}]
    The key idea of the proof is to show the existence of a vertex set $U \subseteq V$ that splits the graph into connected components $C_1,\dots,C_k$ of $G[V \setminus U]$ such that \begin{inparaenum}[(a)] \item $\forall_{i \in [k]} |C_i| = r$, and \item for every $v \in V$ there exists $i \in [k]$ such that $\dist_G(v, C_i) \leq r$.
\end{inparaenum}

Given this decomposition, we encode the string $S$ on each of the connected components $C_1,\dots,C_k$, assigning $O(1)$ bits of $S$ to each node in the component. Therefore, every vertex $v$ can retrieve $S$ from its closest component $C_i$. Moreover, we can guarantee that, as long as $G$ is connected, all nodes output the same string. To achieve this, we define an auxiliary graph $G'$ whose vertices correspond to the clusters, and where two clusters are adjacent in $G'$ if there exists a short path between them in $G$. We then observe that $G'$ is connected, meaning that to ensure all nodes output the same string, it suffices to guarantee that neighboring clusters encode the same string.
\end{proof}

We are now ready to prove \Cref{thm:general_main}. Let $(\Prov_{\mathrm{share}},\Ver_{\mathrm{share}})$ be the labeling scheme for sharing strings described in \Cref{lem:shared_string}. We construct an $O(t)$-PLS $(\Prov_t,\Ver_t)$ as follows.

\paragraph*{Description of $\Prov_t$:} Let $(\cC,X)$ be a $(t,\varepsilon)$-TS partition of $G$. For a randomness string $R:[\mathrm{poly}(n)] \rightarrow [2t\log{n}+2,8t\log{n}]$, let $T_R$ be the following set:
\[T_R = \{v_i \in V \mid C_i \text{ is taken into $\cC$ when running $\cA$ using randomness $R$} \}.\]

Recall that \Cref{lem:general_good_randomnesss} provides a function $F_n:[n^2]\rightarrow \{0,1\}^*$, such that for any graph in $\cG_n$ there exists an index $i \in [n^2]$ such that $F_n(i)$ is good randomness for it. Let $i_G \in [n^2]$ be a good randomness string for running $\cA$ in $G$. Let $R^* = F_n(i_G)$, let $t(v) = 1$ if $v \in T_{R^*}$, and otherwise $t(v) = 0$, and let $x(v) = 1$ if $v \in X$, and otherwise $x(v) = 0$.  

Denote $\var{Shared}(i_G,v) = \Prov_{\mathrm{share}}(G,i_G,v)$, the label given by $\Prov_{\mathrm{share}}$ to node $v$ when sharing the binary encoding of $i_G$ to all nodes. The prover gives the following label to each node $v \in V$.

\[\Prov_t(G,I,v) = \left(\var{Shared}(i_G,v),t(v),x(v) \right).\]

By \Cref{lem:shared_string}, the cost of $\Prov_{\mathrm{share}}$ is $O(1)$ bits, and the size of both $t(v)$ and $x(v)$ is a single bit. Therefore in total, the cost of $\Prov_t$ is $O(1)$ bits. 

\paragraph*{Description of $\Ver_t$:} For convenience, we describe the verifier as a distributed algorithm from the viewpoint of $v$. We assume that the proof labels of any node $v$ are given in a format encoding three strings $(\var{ShareLabel}(v),\var{InT}(v),\var{InX}(v))$, where $\var{InT}(v),\var{InX}(v)$ are single bits, and $\var{ShareLabel}(v)$ is a constant-sized string. Otherwise we have $\Ver_t$ reject.

Before describing the verifier, we define a simple sub-procedure called $\mathrm{FindMyCluster}(v,R)$, running from the viewpoint of node $v$. Node $v$ is given as input a function $R:[\mathrm{poly}(n)] \rightarrow [2t\log{n}+2,8t\log{n}]$ mapping every possible identifier to an even index in that range. The sub-procedure returns the smallest identifier of a node $u \in B_{8t\log{n}}(v)$ satisfying both \begin{inparaenum}[(a)]
    \item $\var{InT}(u) = 1$ and 
    \item $\dist_G(u,v) \leq R(\ID(u))$.
\end{inparaenum} If there is no such node, then $v$ outputs $\bot$. Informally, $\mathrm{FindMyCluster}(v,R)$ simulates algorithm $\cA$ from the viewpoint of $v$, assuming randomness string $R$, and assuming $\var{InT}$ is the set of vertices for which a cluster was taken into $\cC$. Its output is the identifier of $v$'s cluster leader, where $\bot$ represents an error. We prove a few simple facts about $\mathrm{FindMyCluster}$.

    \begin{lemma}
    \label{lem:general_same_cluster_short_dist}
        If for nodes $u,v \in V$ and some randomness string $R$ we have $\mathrm{FindMyCluster}(u,R) = \mathrm{FindMyCluster}(v,R) \neq \bot$, then $\dist_G(u,v) \leq 16t\log{n}$.
    \end{lemma}
    \begin{proof}
        Recall that $\mathrm{FindMyCluster}$ is a distributed algorithm that has access to a local neighborhood of radius $8t\log{n}$ and returns an identifier of some node in this neighborhood. Let $w$ be a node such that $\ID(w)=\mathrm{FindMyCluster}(u,R)=\mathrm{FindMyCluster}(v,R)$. Then necessarily $w\in B_{8t\log{n}}(u)\cap B_{8t\log{n}}(v)$, and by the triangle inequality we have $\dist_G(u,v) \leq \dist_G(u,w)+\dist_G(w,v) \leq 16t\log{n}$.
    \end{proof}

    We next formalize the fact that $\mathrm{FindMyCluster}$ outputs the lexicographically smallest ``taken'' center whose ball contains the node.

    \begin{lemma}
\label{lem:findmycluster_min_taken}
Fix a randomness string $R$ and labeling $\var{InT}:V\to\{0,1\}$. For every node $u \in V$,
\[\mathrm{FindMyCluster}(u,R) = \min\left\{\ID(w) \mid\ w \in B_{8t\log{n}}(u),\ \var{InT}(w)=1,\ \dist_G(w,u) \le R(\ID(w))\right\},
\]
where $\mathrm{FindMyCluster}(u,R) = \bot$ if the set is empty.
\end{lemma}
\begin{proof}
This follows by definition of $\mathrm{FindMyCluster}$, which returns the smallest identifier among nodes satisfying the two conditions above.
\end{proof}

We are now ready to describe $\Ver_t$:

\begin{enumerate}[(1)]
    \item Node $v$ runs $\Ver_{\mathrm{share}}$, using $\var{ShareLabel}$ as the labeling and rejects if it rejects. Otherwise, node $v$ obtains some string $i'(v)$, which it interprets as an encoding of a number. By the soundness of $(\Prov_{\mathrm{share}},\Ver_{\mathrm{share}})$, we can assume that unless some node rejects, $i' = i'(v)$ is the same for all nodes, and omit $v$ from the notation for convenience. Node $v$ computes $F_{n}(i') = \widetilde{R}$ where $\widetilde{R}:[\mathrm{poly}(n)] \rightarrow [2t\log{n}+2,8t\log{n}]$.

    \item \label{item:general_ver_findmycluster} Node $v$ rejects if $\mathrm{FindMyCluster}(v,\widetilde{R}) = \bot$.
    
    \item  For every node $u \in B_{16t\log{n}}(v)$, node $v$ runs $\mathrm{FindMyCluster}(u,\widetilde{R})$, and defines
    \[\widetilde{C}_v(v) = \{u \in B_{16t\log{n}}(v) \mid \mathrm{FindMyCluster}(u,\widetilde{R}) = \mathrm{FindMyCluster}(v,\widetilde{R}) \}.\]
    Moreover, for any $w \in B_2(\widetilde{C}_v(v))$ it defines
    \[\widetilde{C}_w(v) = \{u \in B_{16t\log{n}}(w) \mid \mathrm{FindMyCluster}(u,\widetilde{R}) = \mathrm{FindMyCluster}(w,\widetilde{R}) \}.\]
    \item $v$ adds $\widetilde{C}_w(v) \in \Comp(v)$ for every $w \in B_2(\widetilde{C}_v(v))$, and defines for each such $C \in \Comp(v)$ the set $X_C = \{u \in C \mid \var{InX}(u) = 1\}$.
    \item \label{item:general_ver_allprop}Node $v$ asserts that $\widetilde{C}_v(v)$ has weak diameter $\leq 16t\log{n}$, satisfies the cost ratio property, and that any path from $\widetilde{C}_v(v) \setminus X_{\widetilde{C}_v(v)}$ to any node $u \in V \setminus \widetilde{C}_v(v)$ with $\var{InX}(u) = 0$ passes through at least two consecutive nodes $z_1,z_2$ with $\var{InX}(z_1) = \var{InX}(z_2) = 1$.
\end{enumerate}

    \paragraph*{Correctness:} Assume every node $v$ receives the label $\Prov_t(G,I,v)$. We prove that all nodes accept, and output their corresponding part in the TS partition $(\cC,X)$ obtained by running algorithm $\cA$ with randomness string $R^*$.

\begin{lemma}
    Let $C_u$ be the cluster node $u$ belongs to at the end of an execution of algorithm $\cA$ with randomness string $R^*$. $\widetilde{C}_v(v) = C_v$, and for every $u \in B_2(C_v)$ we have $\widetilde{C}_u(v) = C_u$. 
\end{lemma}
\begin{proof}
    Let $v_1,\dots,v_n$ be the nodes of the graph, ordered lexicographically. Recall that 
    \[T_R = \{v_i \in V \mid C_i \text{ is taken into $\cC$ when running $\cA$ using randomness $R$} \},\]
    and that $\var{InT}(v) = 1$ if and only if $v \in T$. We prove by induction on $i \geq 1$ that if $v_i \in T$, then for all $u \in C_{v_i}$, we have     
    $\mathrm{FindMyCluster}(u,R^*) = \ID(v_i).$
    Let us consider the case of $i = 1$. If $v_1 \notin T$, then the claim holds trivially. Otherwise, the algorithm $\cA$ sets $C_{v_1} = B_{R^*(v_1)}(v_1)$. Additionally, we have $\var{InT}(v_1) = 1$, and that $v_1$ is the smallest identifier in the network. Therefore for a node $u \in V$ it holds that $\mathrm{FindMyCluster}(u,R^*) = \ID(v_1)$ if and only if $u \in B_{R^*(v_1)}(v_1)$, which completes the basis case.

    Next, we assume that the induction claim holds for some $i \in [n-1]$, and we prove the claim for $i+1$. By the induction assumption, we have that for all $j \leq i$ that for any $u \in C_{v_j}$ that $\mathrm{FindMyCluster}(u,R^*) < \ID(v_{i+1})$. Therefore, the alive set satisfies
    \[\alive_i = \{u \in V \mid \mathrm{FindMyCluster}(u,R^*) \geq \ID(v_{i+1})\}.\]
    Again, if $v_{i+1} \notin T_{R^*}$, then the claim trivially holds. Otherwise, we have that $C_{v_{i+1}}$ is taken into $\cC$ in algorithm $\cA$, and recall that $C_{v_{i+1}} = B_{R^*(v_{i+1})}(v_{i+1}) \cap \alive_i$. 
    
    On the other hand, by \Cref{lem:findmycluster_min_taken}, any node $u \in \alive_i$ such that $\mathrm{FindMyCluster}(u,R^*) \geq \ID(v_{i+1})$ is equivalent to the statement that $u$ does not have a node $w \in B_{8t\log{n}}(u)$ such that $u \in B_{R^*(w)}(w)$, and that $\var{InT}(w) = 1$. But since $\var{InT}(v_{i+1}) = 1$, then for any $u \in B_{R^*(v_{i+1})}(v_{i+1})$ we have $v_{i+1} \in B_{8t\log{n}}(u)$ , implying that $\mathrm{FindMyCluster}(u,R^*) = \ID(v_{i+1})$. This concludes the proof by induction.

    The lemma immediately follows, since $\widetilde{C}_v(v)$ every $\widetilde{C}_w(v)$ for each $w \in B_2(\widetilde{C}_v(v))$ are defined by $\mathrm{FindMyCluster}$ as
    \[\widetilde{C}_v(v) = \{u \in B_{16t\log{n}}(v) \mid \mathrm{FindMyCluster}(u,\widetilde{R}) = \mathrm{FindMyCluster}(v,\widetilde{R}) \},\]
    and
    \[\widetilde{C}_w(v) = \{u \in B_{16t\log{n}}(w) \mid \mathrm{FindMyCluster}(u,\widetilde{R}) = \mathrm{FindMyCluster}(w,\widetilde{R}) \}.\]
    In other words $\widetilde{C}_v(v) = C_v$, and $\widetilde{C}_w(v) = C_w$ for each $w \in B_2(C_v)$.

\end{proof}

Recall that the prover set $\var{InX}(u)$ if and only if $u \in X$, or in other words $X = \{u \in V \mid \var{InX}(u) = 1\}$, therefore it follows that for any $C \in \Comp(v)$ we have $X_C(v) = C \cap X$. By the fact that $R^*$ is a good random string, we are guaranteed that $(\cC,X)$ is an $O((t\log{n}),1/t)$-TS partition, and therefore all clusters satisfy the weak-diameter, cost ratio, and two separation properties. Therefore, all nodes accept, and correctness follows.

\paragraph*{Soundness:}
Before proving soundness, we first prove a few simple claims. We define for every $w \in V$ the cluster
\[\widetilde{C}_w = \{u \in V \mid \mathrm{FindMyCluster}(u,\widetilde{R}) = \mathrm{FindMyCluster}(w,\widetilde{R}) \}.\]

We define $\widetilde{\cC} = \{\widetilde{C}_w \mid w \in V\}$. Additionally, define $\widetilde{X} = \{w \in V \mid \var{InX}(w) = 1\}$ and for each $C \in \widetilde{\cC}$ define $\widetilde{X}_C = C \cap \widetilde{X}$. Our overall goal is to show that the network outputs $(\widetilde{\cC},\widetilde{X})$ and that it is indeed an $(O(t\log{n}),1/t)$-TS partition. We first show that $\widetilde{C}_w = \widetilde{C}_w(v)$ for any $v \in V$.

\begin{lemma}
\label{lem:general_independent_of_v}
    For any $v,w \in V$, we have $\widetilde{C}_w = \widetilde{C}_w(v)$ assuming $\widetilde{C}_w(v)$ is defined by $v$.
\end{lemma}
\begin{proof}
    Recall that
    \[\widetilde{C}_w(v) = \{u \in B_{16t\log{n}}(w) \mid \mathrm{FindMyCluster}(u,\widetilde{R}) = \mathrm{FindMyCluster}(w,\widetilde{R}) \}.\]
    Notice that this definition is completely independent of $v$. Clearly $\widetilde{C}_w(v) \subseteq \widetilde{C}_w$. We show the other direction. Let $u \in \widetilde{C}_w$. Then $\mathrm{FindMyCluster}(u,\widetilde{R}) = \mathrm{FindMyCluster}(w,\widetilde{R})$. By \Cref{lem:general_same_cluster_short_dist} we have $\dist_G(u,w) \leq 16t\log{n}$, i.e. $u \in \widetilde{C}_w(v)$ and the claim follows.
\end{proof}

The following follows as an immediate corollary.

\begin{corollary}
Assume all nodes accept. Then for each $v \in V$, $\widetilde{C}_v \in \Comp(v)$ and for each $w \in B_2(\widetilde{C}_v)$ we have $\widetilde{C}_w \in \Comp(v)$. Moreover, for each $C \in \Comp(v)$ we have $X_C(v) = C \cap \widetilde{X}$.
\end{corollary}
\begin{proof}
    By \Cref{lem:general_independent_of_v} we have $\widetilde{C}_v(v) = \widetilde{C}_v$. Therefore $\widetilde{C}_v \in \Comp(v)$ and for each $w \in B_2(\widetilde{C}_v)$ we have $\widetilde{C}_w = \widetilde{C}_w(v) \in \Comp(v)$. The latter claim follows by definition of $\widetilde{X} = \{u \mid \var{InX}(u) = 1\}$, and definition of $X_C(v) = \{u \in C \mid \var{InX}(u) = 1\}$.
\end{proof}

Finally, we show that the pair $(\widetilde{\cC},\widetilde{X})$ is an $(O(t\log{n}),1/t)$-TS partition of $G$ 
\begin{lemma}
    Assuming all nodes accept, the family of subsets $\widetilde{\cC} = \{\widetilde{C}_u \mid u \in V\}$ is a partition of $V$. Moreover, $(\widetilde{\cC},\widetilde{X})$ is an $(O(t\log{n}),1/t)$-TS partition of $G$.
\end{lemma}
\begin{proof}
    Let $u \in V$. Clearly $u \in \widetilde{C}_u$. We show that $u$ appears in exactly one cluster. Indeed, assume that $u \in \widetilde{C}_w$ for some $w \in V$. Then by definition we have 
    \[\mathrm{FindMyCluster}(u,\widetilde{R}) = \mathrm{FindMyCluster}(w,\widetilde{R}).\]
    
    And it follows that $\widetilde{C}_u = \widetilde{C}_w$. Finally, we show that $(\widetilde{\cC},\widetilde{X})$ is an $(O(t\log{n}),1/t)$-TS partition. Since $u$ does not reject, then by Step~(\ref{item:general_ver_findmycluster}) of $\Ver_t$ we have $\mathrm{FindMyCluster}(u,\widetilde{R}) \neq \bot$, and by Step~(\ref{item:general_ver_allprop}) we have that $\widetilde{C}_u(u)$ satisfies the low diameter, two separation, and cost ratio properties.
\end{proof}

We conclude that if all nodes accept, then $(\widetilde{\cC},\widetilde{X})$ is a $(O(t\log{n}),1/t)$-TS partition of $G$, and each node $v$ outputs $\Comp(v)=\{\widetilde{C}_u \mid u \in B_2(\widetilde{C}_v)\}$, with $X_C(v)=C\cap\widetilde{X}$ for every $C\in\Comp(v)$, and soundness follows.  By combining this proof labeling scheme with \Cref{lem:tstradeoff}, we obtain \Cref{thm:general_main}.

\paragraph*{Discussion on knowledge of $n$:} In this section, we assume the verifier has knowledge of $n$ when locally constructing the function $F_n$ of \Cref{lem:general_good_randomnesss}, that is used to obtain our randomness strings. However, we believe this assumption can be removed with relative ease by using the string sharing PLS of \Cref{lem:shared_string} to share the value $n'$ in addition to  the index $i'$. We note that while $n'$ might not be the real value of $n$, we are still guaranteed that the generated randomness string is identical among all nodes. Moreover, since the nodes verify the properties of the resulting TS partition, if they all accept, then the TS partition is valid even if $n'$ is not the real size of the graph. Once the assumption of knowing $n$ is removed, we can also remove the assumption that the graph is connected, since we can act on each component separately.

\section{Refuting Stronger Variants of the Tradeoff Conjecture}
\label{sec:lower}

In this section, we refute a stronger variant of the Tradeoff Conjecture, and show a problem where every node has $|B_t(v)| = \Omega(n)$, but we cannot get an improvement of better than $O(1/t)$ when considering a $t$-PLS compared to a $1$-PLS. 

\thmdisproof*

In our proof, we use a reduction to 2-party non-deterministic communication complexity. We begin with a brief overview of the 2-party non-deterministic communication model and the results relevant for our reduction. For simplicity, instead of describing the standard non-deterministic communication model, we describe a slight (but standard) variant. This variant is equivalent to the standard model up to an additive $O(1)$ term in their respective complexity measures.

\paragraph*{Non-deterministic communication complexity:}
For our purposes, the non-deterministic 2-party communication model consists of two players, Alice and Bob, with private inputs $X \in \cX$ and $Y \in \cY$, respectively, for some input domain $\cX \times \cY$. A protocol with cost $p$ in this setting is defined as follows. Initially, both Alice and Bob are given a shared string $W \in \{0,1\}^{\leq p}$, called the \emph{witness string}. Without communicating, both Alice and Bob output either \emph{accept} or \emph{reject} based solely on their private input and the witness $W$.  For a predicate $P : \cX \times \cY \rightarrow \{0,1\}$, we say that a non-deterministic $p$-cost protocol solves $P$ if:

\begin{itemize}

\item \textbf{Correctness:} For any $(X,Y) \in \cX \times \cY$ such that $P(X,Y) = 1$, there exists a witness string $W$ with $|W| \leq p$ for which both players accept.

\item \textbf{Soundness:} For any $(X,Y) \in \cX \times \cY$ such that $P(X,Y) = 0$, and for every witness string $W$ with $|W| \leq p$, at least one of Alice or Bob rejects.

\end{itemize}

We briefly remark that in the standard model, nodes are allowed to communicate, and the cost is the size of $W$ plus the total number of bits communicated. Though, using a standard argument, the difference in cost between the two models is at most two bits. Informally, this is because instead of having lengthy communication, we can always encode the transcript inside $W$ and have both nodes verify the transcript using a single bit each (See \cite{KN97} for an overview).

Finally, we describe the equality function and the lower bound needed for our results. The equality function $\mathrm{EQ}_k : \{0,1\}^k \times \{0,1\}^k \to \{0,1\}$ is defined by $\mathrm{EQ}_k(X,Y) = 1$ if and only if $X = Y$.

\begin{lemma}[\cite{KN97} Chapter 2.1]
\label{lem:non_deterministic_eq}
    Any non-deterministic 2-player protocol solving $\mathrm{EQ}_k$ has cost $\Omega(k)$.
\end{lemma}

\paragraph*{Proof of \Cref{thm:disproof_strong}:}
Next, we prove \Cref{thm:disproof_strong}. We begin by describing $G_{t,m}$ and $P$, and continue to show the existence of an $O(m)$ cost $1$-PLS for $P$, and the lower bound of $\Omega(m/t)$ on the cost of any $t$-PLS for $P$.  The graph $G_{t,m}$ is defined as follows for integer $t \geq 1$, and odd integer $m \geq 3$. For $i \in [2t+3]$, we define the set of vertices, \[V_i = 
       \begin{cases}
             \{v^i\},  & i \text{ is odd} , \\
            \{v^i_1,\dots,v^i_m\}, & i \text{ is even}.
        \end{cases}\] 

Let $V = \bigcup_{i \in [2t+3]} V_i$ be a union of the vertex sets, and note that $|V| = \Theta(mt)$. Let $G_{t,m} = (V,E)$ be a graph where $E = \{\{u,w\} \mid i \in[2t+2], u \in V_i, w \in V_{i+1}\}$, i.e. $G_{t,m}$ contains for every $i \in [2t+2]$ all edges between $V_i$ and $V_{i+1}$ (See \Cref{fig:lowerbound}). 

We define the problem $P$ as follows. In this problem, both nodes $v^1 \in V_1, v^{2t+3} \in V_{2t+3}$ receive an input $I(u) \in \{0,1\}^{m^2}$. We ignore all additional inputs on nodes not in $V_1 \cup V_{2t+3}$, i.e. for nodes with even degree. We define $P$ so that $(G_{t,m},I) \in P$ if and only if $I(v^1) = I(v^{2t+3})$. Denote $\deg(v)$ as the degree of a vertex $v$. We recall that $m$ is an odd integer, and note the following trivial properties.

\begin{lemma}
    The following holds for any $v \in V$.
    \begin{enumerate}
        \item If $\deg(v) = 2$, then $v$ is at an even layer.
        \item If $\deg(v) > 2$ and is odd, then $v \in V_1 \cup V_{2t+3}$. Moreover, $\deg(v) = m$.
        \item If $\deg(v) > 2$ and is even, then $v$ is at an odd layer, but not in $V_1 \cup V_{2t+3}$. Moreover, $\deg(v) = 2m$.
    \end{enumerate}
\end{lemma}

\begin{lemma}
    There exists a $1$-PLS for $P$ in $G_{t,m}$ with cost $O(m)$.
\end{lemma}
\begin{proof}
We describe a $1$-PLS for $P$, denoted $(\Prov_1,\Ver_1)$, with cost $O(m)$.
    \paragraph*{Description of $\Prov_1$}
    Let $I_1,\dots,I_m$ be a partition of $I(v^1)$ into $m$ equal-sized segments. For any $j \in [2t+3]$, if $j$ is odd, then $\Prov_1(G,I,v) = \phi$, i.e. an empty label. If $j$ is even, then for $i \in [m]$, node $v_i^j$ is given the label we set $\Prov_1(G,I,v_i^j) = (i,I_i)$. We remark that indeed all label sizes are of size at most $O(m)$, hence the cost is at most $O(m)$.
    \paragraph*{Description of $\Ver_1$} We define $\Ver_1$ as a distributed protocol from the viewpoint of $v$:
    \begin{enumerate}
        \item If $\deg(v) = 2$ (i.e. $v$ is at an even layer), $v$ accepts. 
        \item If $\deg(v) > 2$ and is even (i.e. $v$ is at an odd layer, but not in $V_1 \cup V_{2t+3}$), then $v$ sets $m = \deg(v)/2$, and asserts that for each $i \in [m]$ there are exactly two labels of neighbors $(i_1,S_1)$ and $(i_2,S_2)$ such that $i = i_1 = i_2$ and $S_1 = S_2$.
        \item If $\deg(v) > 2$ and is odd (i.e. $v \in V_1 \cup V_{2t+3}$), $v$ asserts that for each $i \in [m]$ there is exactly one label of neighbors $(i_1,S_1)$ where $i = i_1$, and that its input is equal to the concatenation of strings $I(v) = S_1 \circ \dots \circ S_m$.
    \end{enumerate}
    
    \paragraph*{Correctness}~ We split into cases by layer. Every odd layer vertex always accepts. For each of the layers $V_1,V_{2t+3}$ we see that indeed for their unique node $v$, they have $m$ neighbors, each with label $(i,I_i)$, and indeed $I(v) = I_1 \circ \dots \circ I_m$.
    Similarly, for the rest of the odd layers, each node receives for each $i \in [m]$ the label $(i,I_i)$ exactly twice from its neighbors, and we conclude that all nodes accept.

    \paragraph*{Soundness}~Assume all nodes accept. We prove by induction over all even integers in $j \in [2t+3]$ that for all $i \in [m]$ there is exactly one node $u_{i,j} \in V_j$ that has the label $(i,I_i)$. For the case $j = 2$ this must hold, otherwise $v^1$ rejects. Assume this for even layer $j \geq 2$, therefore layer $j+2$ must also have exactly one node $u_{i,j+2} \in V_{j+2}$ that has the label $(i,I_i)$, otherwise $v^{j+1}$ rejects. Finally, we must have
    $I(v^{2t+3}) = I_1 \circ \dots \circ I_m$,
    otherwise $v^{2t+3}$ rejects. We conclude that $I(v^1) = I(v^{2t+3})$.
    
\end{proof}

Finally, we show that any $t$-PLS certifying $P$ in $G_{t,m}$ has cost $\Omega(m/t)$, which is shown by reduction to non-deterministic 2-player equality.

\begin{lemma}
    Any $t$-PLS certifying $P$ in $G_{t,m}$ has cost $\Omega(m/t)$.
\end{lemma}  
\begin{proof}
     Let $\Pi_t = (\Prov_t,\Ver_t)$ be a $t$-PLS for certifying $P$ in $G_{t,m}$ with cost $p$. We show that this implies a non-deterministic 2-player protocol for $(m^2)$-bit equality with cost $O(p \cdot m \cdot t)$. Combining with the $\Omega(m^2)$ complexity lower bound on $\mathrm{EQ}_{m^2}$ (\Cref{lem:non_deterministic_eq}), we obtain by rearranging that $p = \Omega(m/t)$.

We describe a 2-player protocol that uses $\Pi_t$ to obtain the complexity mentioned above. Let $X,Y$ be the private inputs of Alice and Bob respectively. Alice and Bob each construct the graph $G_{t,m}$. Alice assigns $X$ to be the input of $v^1$, and similarly, Bob assigns $Y$ to be the input of $v^{2t+3}$. Both players interpret $W = (\ell_1,\dots,\ell_{|V(G)|})$ as a labeling on the nodes $V$ (rejecting if $W$ does not encode $|V(G)|$ labels). Finally, Alice checks whether all nodes in $\bigcup_{i=1}^{t+1} V_{i}$ accept, and Bob checks whether all nodes of $\bigcup_{i=t+2}^{2t+3} V_i$ accept. If so, they accept, and otherwise reject. 

We analyze the correctness of the protocol above. Assuming $X = Y$, we show there exists a witness string $W$ causing both players to accept. Indeed, taking $W = (\Prov_t(G,I,v))_{v \in V}$, and noting that $X=Y$ implies $I(v^1) = I(v^{2t+3})$, then by correctness of $\Pi_t$ all nodes accept.  Otherwise, assume that $X \neq Y$, i.e. $I(v^1) \neq I(v^{2t+3})$. Therefore, by the soundness of $\Pi_t$, for any labeling $\ell:V\rightarrow\{0,1\}^*$ the verifier $\Ver_t$ rejects on at least one node $v$. We note that either $v \in \bigcup_{i=1}^{t+1} V_i$ or $v \in \bigcup_{i=t+2}^{2t+3} V_i$. In the first case, Alice rejects, and in the latter Bob rejects.  
\end{proof}

\section*{Acknowledgments}
This work is supported in part by the Israel Science Foundation, grant No. 1042/22 and 800/22.

\bibliographystyle{alpha}
\bibliography{bibliography}

\appendix
\section{Missing Proofs from \texorpdfstring{\Cref{sec:minor_free_graphs}}{Section~\ref{sec:minor_free_graphs}}}
\label{sec:app_minor}

In this section, we prove \Cref{lem:tsminorcompute}. Namely, we show that any $(t,\varepsilon)$-TS partition can be certified and outputted in any graph  using a $O(t)$-PLS with cost $O(\log{n})$.

\tsminorcompute*

     Let $r = 3t+2$. We construct an $r$-PLS $(\Prov,\Ver)$ certifying and outputting a $(t,\varepsilon)$-TS partition as follows.

    \paragraph*{Description of $\Prov$:} Let $(\cC,X)$ be a $(t,\varepsilon)$-TS partition of $G$. For each cluster $C \in \cC$, the honest prover assigns a unique identifier $\var{ID}(C) \in [n]$, and gives node $v$ the label \[\Prov(G,I,v) = (\var{ID}(C_v),x(v)),\] where $x(v)$ is a single bit set to $x(v) = 1$ if $v \in X$, and $x(v) = 0$ otherwise. We note that the size of the label for each node is $O(\log{n})$ bits.
    
    \paragraph*{Description of $\Ver$:}  For convenience, we describe $\Ver$ as a distributed algorithm running from the viewpoint of node $v$. We assume that $v$ receives for every node $u \in B_r(v)$ a label $(\var{CompID}(u),\var{InX}(u))$, where $\var{InX}(u)$ is a single bit (otherwise $v$ rejects). Given the $r$-hop topology of $v$, node $v$ defines 
    \[\widetilde{C}_v(v) = \{w \in B_r(v) \mid \var{CompID}(v) = \var{CompID}(w)\},\]
    and for each $u \in B_{2}(\widetilde{C}_v(v))$ defines the cluster \[\widetilde{C}_u(v) = \{w \in B_{r}(u) \mid \var{CompID}(u) = \var{CompID}(w)\}.\] 
    It then defines $\Comp(v) = \{\widetilde{C}_u(v) \mid u \in B_2(\widetilde{C}_v(v))\}$, and for each $C \in \Comp(v)$ defines the set $X_C(v) = \{w \in C \mid \var{InX}(w) = 1\}$.
    
     Node $v$ asserts the following conditions, and rejects if at least one assertion fails.
    \begin{enumerate}[(1)]
        \item \label{item:log_ts_assert_1} Node $v$ asserts that for each $C \in \Comp(v)$ that $\weakdiam(C) \leq t$, and the cost ratio is at most $|X_C(v)|/|C| \leq \varepsilon$.
        \item \label{item:log_ts_assert_2} Node $v$ asserts that there are no two nodes $u \in \widetilde{C}_v(v)$ and $w \notin \widetilde{C}_v(v)$, with $\var{InX}(u) = \var{InX}(w) = 0$, such that $\dist_G(u,w) \leq 2$.
    \end{enumerate} 
    If all assertions are true, then $\Ver$ accepts and outputs $\Comp(v)$, and for each $C \in \Comp(v)$ outputs the set $X_{C}(v)$ as its corresponding $X$-set. This concludes the description of the verifier. 
    
    \paragraph*{Correctness:} Since $(\cC,X)$ is a $(t,\varepsilon)$-TS partition, and each $C \in \cC$ is given a distinct identifier, then $\widetilde{C}_v(v) = C_v$, and for any $u \in B_2(C_v)$, it similarly holds that $\widetilde{C}_u(v) = C_u$. Moreover, for each $u \in B_2(C_v)$, the cluster $\widetilde{C}_u(v) = C_u$ satisfies the weak-diameter, two separation, and cost ratio conditions. Therefore, all assertions of $\Ver$ running on $v$ are true, and it accepts and outputs $\Comp(v) = \{C_u \mid u \in B_2(C_v)\}$, and for each $C \in \Comp(v)$ the set $X_C(v) = X_C$.
    
    \paragraph*{Soundness:} Assume that for every node $v \in V$, $\Ver$ accepts for some labeling $\ell$. We begin by showing a few simple claims. 

    \begin{lemma}
        \label{lem:logn_cert_simple} If $u \in B_r(v)$, and $u \notin \widetilde{C}_v(v)$, then $\var{CompID}(u) \neq \var{CompID}(v)$.
    \end{lemma}
    \begin{proof}
        By definition, $\widetilde{C}_v(v) = \{w \in B_{r}(v) \mid \var{CompID}(w) = \var{CompID}(v)\}$. Hence, if $u \in B_r(v)$ and $u \notin \widetilde{C}_v(v)$, then necessarily $\var{CompID}(u) \neq \var{CompID}(v)$.
    \end{proof}

    \begin{lemma}
    \label{lem:clusters_consistent}
        Assume all nodes accept. For any $u,v,w \in V$, if $\widetilde{C}_w(u) \in \Comp(u)$ and $\widetilde{C}_w(v)  \in \Comp(v)$, then $\widetilde{C}_w(u) = \widetilde{C}_w(v)$. Moreover, in that case $X_{\widetilde{C}_w}(u) = X_{\widetilde{C}_w}(v)$.
    \end{lemma}
    \begin{proof}    
        Assume by contradiction that there exists a node $z \in \widetilde{C}_w(v) \setminus \widetilde{C}_w(u)$. Since $z \in \widetilde{C}_w(v)$, by definition we have $z \in B_r(w)$ and $\var{CompID}(z) = \var{CompID}(w)$. On the other hand, since $z \notin \widetilde{C}_w(u)$ while $\var{CompID}(z) = \var{CompID}(w)$, again by definition of $\widetilde{C}_w(u)$ it must be that $z \notin B_r(w)$, a contradiction. Therefore $\widetilde{C}_w(v) = \widetilde{C}_w(u)$.

        Moreover, the corresponding sets $X_{\widetilde{C}_w}(u)$ and $X_{\widetilde{C}_w}(v)$, defined as the set of vertices $y \in \widetilde{C}_w$  with $\var{InX}(y)=1$, are equal as well.
\end{proof}

    \Cref{lem:clusters_consistent} shows that given a labeling $\ell$ where all nodes accept, the clusters $\widetilde{C}_u(v)$ are consistent whenever they are defined. In the remainder of the section, we denote $\widetilde{C}_u = \widetilde{C}_u(u)$, and $X_{\widetilde{C}_{u}} = X_{\widetilde{C}_u}(u)$. We note that this is defined for all $u \in V$, since we always have $\widetilde{C}_u(u)\in \Comp(u)$. We let $\widetilde{\cC} = \{\widetilde{C}_u \mid u \in V\}$, and $\widetilde{X} = \bigcup X_{\widetilde{C}_{u}}$. We show that $(\widetilde{\cC},\widetilde{X})$ is an $(t,\varepsilon)$-TS partition.
    
    \begin{lemma}
    \label{lem:simple_ts_partition_good}
        Assuming all nodes accept, the family of subsets $\widetilde{\cC} = \{\widetilde{C}_u \mid u \in V\}$ is a partition of $V$. Moreover, $(\widetilde{\cC},\widetilde{X})$ is an $(t,\varepsilon)$-TS partition of $G$.
    \end{lemma}
    \begin{proof}
        To show that $\widetilde{\cC}$ is a partition, we show that every node appears in some cluster, and that any two clusters that the clusters are disjoint. The fact that every $u \in V$ appears in some cluster follows from $u \in \widetilde{C}_u$ and $\widetilde{C}_u \in \widetilde{\cC}$. 

        Assume that $z \in \widetilde{C}_u \cap \widetilde{C}_v$. Then $\var{CompID}(u) = \var{CompID}(z) = \var{CompID}(v)$. Since $\widetilde{C}_u \in \Comp(u)$ and $\widetilde{C}_v \in \Comp(v)$, both $u$ and $v$ accept Assertion~(\ref{item:log_ts_assert_1}), and hence $\weakdiam(\widetilde{C}_u) \le t$ and $\weakdiam(\widetilde{C}_v) \le t$. It follows that $\dist_G(u,v) \le 2t$. We claim that $\widetilde{C}_u = \widetilde{C}_v$. Indeed, let $y \in \widetilde{C}_u$. Then $\var{CompID}(y) = \var{CompID}(u) = \var{CompID}(v)$. Moreover, since $\weakdiam(\widetilde{C}_u) \le t$, we have $\dist_G(u,y) \le t$, and thus by the triangle inequality $\dist_G(v,y) \le \dist_G(v,u)+\dist_G(u,y) \le 2t+t = 3t \le r$. Therefore $y \in B_r(v)$ and has the same component identifier as $v$, i.e., $y \in \widetilde{C}_v$. The symmetric argument shows $\widetilde{C}_v \subseteq \widetilde{C}_u$, concluding $\widetilde{C}_u = \widetilde{C}_v$.
        
        We conclude that since every node passes Assertions~(\ref{item:log_ts_assert_1}),(\ref{item:log_ts_assert_2}), then all clusters $\widetilde{C}_v$ have the weak diameter, cost ratio, and two separation properties. In other words, the pair $(\widetilde{\cC},\widetilde{X})$ is a $(t,\varepsilon)$-TS partition.
    \end{proof}

    Finally, we show that all nodes output their respective local part of $(\widetilde{\cC},\widetilde{X})$, and \Cref{lem:tsminorcompute} follows.
    
    \begin{lemma}
        Assume all nodes accept. Then for each $v \in V$, $\widetilde{C}_v \in \Comp(v)$ and for each $u \in B_2(\widetilde{C}_v)$ we have $\widetilde{C}_u \in \Comp(v)$. Moreover, for each $C \in \Comp(v)$ we have $X_C(v) = C \cap \widetilde{X}$.
    \end{lemma}
    \begin{proof}
        By definition $\widetilde{C}_v = \widetilde{C}_v(v)$ and indeed is defined by $v$ and added to $\Comp(v)$. Therefore, $B_2(\widetilde{C}_v) = B_2(\widetilde{C}_v(v))$. In particular, $v$ defines for any $u \in B_2(\widetilde{C}_v)$ a cluster $\widetilde{C}_u(v)$ and adds it to $\Comp(v)$. By \Cref{lem:clusters_consistent}, we have $\widetilde{C}_u(v) = \widetilde{C}_u$, and the first part of the claim follows.

        For latter claim, let $C \in \Comp(v)$. By definition of $\Comp(v)$, there exists $u \in B_2(\widetilde{C}_v)$ such that $C = \widetilde{C}_u(v)$. By \Cref{lem:clusters_consistent}, $C = \widetilde{C}_u$ and moreover $X_C(v) = X_{\widetilde{C}_u}$. Since $\widetilde{\cC}$ is a partition of $V$ (\Cref{lem:simple_ts_partition_good}), we have $X_{\widetilde{C}_u} = \widetilde{C}_u \cap \widetilde{X} = C \cap \widetilde{X}$.
    \end{proof}

\section{Missing Proofs from \texorpdfstring{\Cref{sec:general_graphs}}{Section~\ref{sec:general_graphs}}}
\label{sec:app_general}

In this section, we prove \Cref{lem:shared_string}. We start by first proving the following decomposition lemma.

\begin{lemma}
    \label{lem:shared_string_decomp}
    Let $r \geq 1$ be an integer parameter, and let $G = (V,E)$ be a connected graph with $|V| > r$. There exists a vertex set $U \subseteq V$ that satisfies the following properties. Let $C_1,\dots,C_k$ be the connected components of $G[V \setminus U]$. Then,
    \begin{inparaenum}[(a)]
        \item \label{item:randomness_1} For every $i \in [k]$, $|C_i| = r$, and
        \item \label{item:randomness_2} For every $v \in V$, there exists $i \in [k]$ such that $v \in B_r(C_i)$.
    \end{inparaenum} 
\end{lemma}
\begin{proof}
    We use the following simple iterative procedure to construct the set $U$ and clusters $C_1,\dots,C_k$, proving this decomposition's existence. At each step we have our set of active nodes $V_i$, initially set at $V_1 = V$. At the $i$-th step, if there is no connected component remaining in the graph with size at least $r$, we terminate the process. Otherwise, set $C_i$ to be an arbitrary connected subgraph of such a component, of size exactly $r$, and let
    \[U_i = \{u \in V_i \setminus C_i \mid u \in B_1(C_i)\}.\]
    We set $V_{i+1} = V_i \setminus (C_i \cup U_i)$. Assuming the process ended at step $k+1$, we set $U = V_{k+1} \cup \bigcup_{j=1}^{k} U_j$.

    Next we prove correctness of this process. 
    \begin{itemize}
        \item Property~(\ref{item:randomness_1}) follows trivially from the fact that each cluster is taken to be a connected subgraph of size exactly $r$.
        \item For Property~(\ref{item:randomness_2}), we consider a few simple cases. If $v \in C_i \cup U_i$ for some $i \in [k]$, then $v \in B_1(C_i)$. Otherwise, $v \in V_{k+1}$, i.e. $v$ is one of the nodes remaining at the end of the process. Since $k+1$ is the last step, the connected component of $v$ in $V_{k+1}$ is of size at most $r-1$. In particular, $\dist_G(v,\cup_{i=1}^{k}U_i) \leq r-1$, implying $v \in B_r(C_i)$ for some $i \in [k]$.
    \end{itemize} 
\end{proof}

Next, we show an implication of Property~\Cref{lem:shared_string_decomp}(\ref{item:randomness_2}), which we use for the soundness of the protocol. For an integer parameter $r \geq 1$, and set $U \subseteq V$, define the graph $G'(r,U)$ as follows. Let $C_1,\dots,C_k$ be the maximal connected components of $G[V \setminus U]$. Let $G'(r,U)$ be the graph with vertices $V' = \{v'_1,\dots,v'_k\}$, such that two vertices $v'_i,v'_j \in V'$ are connected by an edge if and only if $\dist_G(C_i,C_j) \leq 2r+1$. 
\begin{lemma}
\label{lem:cluster_graph_is_connected}
    Let $r \geq 1$ be an integer parameter, let $G = (V,E)$ be a connected graph, and let $U \subseteq V$. If $\max_{v \in V} \dist_G(v, V\setminus U) \leq r$, then $G'(r,U)$ is connected.
\end{lemma}
\begin{proof}
    Consider a path $p = (w_1 = u,w_2,w_3,\dots,w_\ell =v)$ between two nodes $u,v \in V \setminus U$. For $j \in [\ell]$, let $C_{i_j}$ denote the maximal connected component in $G[V \setminus U]$ closest to $w_j$ in $G$ (breaking ties arbitrarily). By assumption, $\dist_G(w_{j},C_{i_j}) \leq r$. It follows by the triangle inequality that $\dist_G(C_{i_j},C_{i_{j+1}}) \leq 2r+1$, i.e. $v'_{i_j}$ is connected by an edge to $v'_{i_{j+1}}$ in $G'(r,U)$, for all $1 \leq j \leq \ell-1$. Hence, $v'_i,v'_j$ are connected in $G'(r,U)$ if there is a path between the connected components $C_i$ and $C_j$ in $G$. Since $G$ is connected, we conclude that $G'(r,U)$ is also connected.
\end{proof}

Before the proof of \Cref{lem:shared_string}, recall the functions $\lex$ and $\decodelex$ for encoding and decoding strings across a cluster $C$ that were introduced in \Cref{sec:TSpartition}. In short, $\lex(C,S,v)$ pads $S$ to a multiple of $|C|$, and returns the $i$-th block, where $i$ is the lexicographic rank of $v$ in $C$. Conversely, $\decodelex(C,\ell)$, where $\ell(v) = \lex(C,S,v)$ for $v \in C$, reconstructs the original string $S$ by concatenating these blocks in order and removing the padding.

\paragraph*{Proof of \Cref{lem:shared_string}}
    We show a $(4r+2)$-PLS for the problem described. We assume that the labeling given to a node $v$ is of the form $(\var{InU}(v),\var{PartS}(v))$, where $\var{InU}(v)\in\{0,1\}$ is one bit and $\var{PartS}(v)\in\{0,1\}^*$ is an arbitrary string (otherwise, the verifier rejects). Given a labeling $(\var{InU},\var{PartS})$ on all nodes, let $V_0 = \{v \in V \mid \var{InU}(v) = 0\}$. A $0$-cluster $C \subseteq V$ is a maximal connected component in $G[V_0]$, i.e. a maximal connected component in the graph induced by nodes with $\var{InU}(v) = 0$.
    
    \paragraph*{Description of $\Prov$:} Fix a set $U \subseteq V$ satisfying the properties described in \Cref{lem:shared_string_decomp}. Set $\Prov$ as a function which gives labels $(\var{InU}(v),\var{PartS}(v))$ to each node $v$, where
    \begin{inparaenum}[(a)]
        \item Set $\var{InU}(v) = 1$ if $v \in U$, and $\var{InU}(v) = 0$ otherwise.
        \item If $v \in U$, set $\var{PartS}(v) = \emptyset$. Otherwise, let $C_v$ be the $0$-cluster that contains $v$, and set $\var{PartS}(v) = \lex(C_v,S,v)$.
    \end{inparaenum}

    \paragraph*{Description of $\Ver$:}
    We describe the verifier from the viewpoint of $v$, and say that node $v$ rejects if $\Ver(G,I,v) = \reject$, and node $v$ accepts otherwise. First, $v$ rejects if its label is not in the format $(\var{InU}(v),\var{PartS}(v))$ with $\var{InU}(v)\in\{0,1\}$. Let $\var{SeenClusters}(v) = \{C \subseteq B_{4r+1}(v) \mid C \text{ is a $0$-cluster}\}$. Node $v$ rejects if $|\var{SeenClusters}(v)| = 0$,  if $\min_{C \in \var{SeenClusters}(v)} \dist_G(v,C) > r$, or if there is a $0$-cluster $C$ such that $\dist_G(v,C) \leq r$, but such that $|C| \neq r$. 
    
    Then $v$ runs for all $C \in \var{SeenClusters}(v)$ the function $\decodelex(C,\var{PartS})$ and obtains a string $S'(C)$. It then asserts that for all $C,C' \in \var{SeenClusters}(v)$ we have $S'(C) = S'(C')$. If this assertion is false, $v$ rejects. Otherwise, there is some string $S' = S'(C)$ for all $0$-clusters $C \in \var{SeenClusters}(v)$. $v$ accepts and outputs the string $S'$.

    \paragraph*{Correctness:} We need to show that $\Ver(G,I,\Prov(G,I),v)$ accepts for all $v \in V$, and that all nodes output the string $S$. By \Cref{lem:shared_string_decomp}(\ref{item:randomness_2}) there exists a $0$-cluster $C$ such that $\dist_G(v,C) \le r$, and in particular $C \subseteq B_{4r+1}(v)$, so $|\var{SeenClusters}(v)| \geq 1$ and $\min_{C' \in \var{SeenClusters}(v)} \dist_G(v,C') \le r$. Let $C \in \var{SeenClusters}(v)$. We note that for any $u \in C$ we have $\var{PartS}(u) = \lex(C,S,u)$. Therefore, $\decodelex(C,\var{PartS}) = S$, i.e. $S'(C) = S$ for all $C \in \var{SeenClusters}(v)$. Correctness follows.
        
    \paragraph*{Soundness:} Assume that all nodes $v \in V$ accept on some labeling $\ell = (\var{InU},\var{PartS})$. Let $U(\ell) = \{v \in V \mid \var{InU}(v) = 1\}$. We show that the following holds.
    
    \begin{lemma}
        If $\Ver(G,I,v) = \accept$ for all $v \in V$, then:
        \begin{enumerate}[(1)]
        \item The cluster graph $G'(r,U(\ell))$ is connected.
        \item \label{item:share_clusters_same_size} For any $0$-cluster $C$, we have that $|C| = r$.
        \item Let $C,C'$ be $0$-clusters that are connected by an edge in $G'(r,U(\ell))$, then $S'(C) = S'(C')$.
        \end{enumerate}
    \end{lemma}
    \begin{proof} We show that indeed all the properties above hold.
        \begin{enumerate}[(1)]
            \item Since every node $v$ accepts, it holds that $\min_{C \in \var{SeenClusters}(v)} \dist_G(v,C) \le r$. In particular $\dist_G(v, V \setminus U(\ell)) \le r$ for all $v \in V$, and hence by \Cref{lem:cluster_graph_is_connected} the graph $G'(r,U(\ell))$ is connected.
            \item  Let $v \in C$, then $C \in \var{SeenClusters}(v)$ (otherwise $v$ rejects). This implies that $|C| = r$, since $v$ asserts that $|C| = r$ for every $C \in \var{SeenClusters}(v)$.
            \item Since $C,C'$ are connected by an edge, then by definition $\dist_{G}(C,C') \leq 2r+1$, and by (\ref{item:share_clusters_same_size}) we have $|C| = r$. Therefore, for a vertex $v \in C$ we have that $C,C' \subseteq B_{4r+1}(v)$. In particular, $C,C' \in \var{SeenClusters}(v)$ and $v$ asserts that $S'(C) = S'(C')$.
        \end{enumerate}
    The lemma implies that there exists some string $S'$, such that for every $0$-cluster $C$, $S'(C) = S'$. This follows from the cluster graph $G'(r,U(\ell))$ being connected, and that for every edge $\{C,C'\} \in E(G'(r,U(\ell)))$ we have that $S'(C) = S'(C')$. Moreover, since for any $v \in V$, we have $|\var{SeenClusters}(v)| \geq 1$ (otherwise $v$ rejects), then node $v$ outputs the string $S'$, and soundness follows.        
    \end{proof}

\end{document}

%% file: macros.tex
\theoremstyle{plain}
\newtheorem{theorem}{Theorem}[section]

\newtheorem{lemma}[theorem]{Lemma}

\newtheorem{corollary}[theorem]{Corollary}
\newtheorem{conjecture}[theorem]{Conjecture}
\newtheorem{question}[theorem]{Question}
\newtheorem{definition}{Definition}

\newcommand{\Prov}{\mathrm{Prov}}
\newcommand{\Ver}{\mathrm{Ver}}

\newcommand{\etal}{{\emph{ et al.}~}}

\def\cA{\ensuremath{\mathcal{A}}}  
  
\def\cC{\ensuremath{\mathcal{C}}}

\def\cG{\ensuremath{\mathcal{G}}}

\def\cI{\ensuremath{\mathcal{I}}}

\def\cP{\ensuremath{\mathcal{P}}}

\def\cX{\ensuremath{\mathcal{X}}}
\def\cY{\ensuremath{\mathcal{Y}}}

\def\eps{\varepsilon}

\newcommand{\var}[1]{\mathrm{#1}}
\newcommand{\accept}{\mathrm{accept}}
\newcommand{\reject}{\mathrm{reject}}
\newcommand{\leader}{\mathrm{leader}}

\DeclareMathOperator{\TS}{TS}

\DeclareMathOperator{\dist}{dist}

\DeclareMathOperator{\lex}{Lex}
\DeclareMathOperator{\decodelex}{DecodeLex}

\newcommand{\alive}{L}
\DeclareMathOperator{\Comp}{Comp}

\DeclareMathOperator{\Pad}{Pad}

\DeclareMathOperator{\Bad}{Bad}

\DeclareMathOperator{\ID}{ID}

\newcommand{\weakdiam}{\mathrm{weak\text{-}diam}}

\usepackage[colorinlistoftodos,prependcaption,textsize=tiny]{todonotes}

%% file: abstract.tex
\begin{abstract}
In the $t$-Proof Labeling Scheme model ($t$-PLS model), our goal is to certify that a network of nodes satisfies a given property $P$. A prover assigns a label to each node, and each node decides to accept or reject based on its labeled $t$-hop neighborhood. If $P$ holds, there exists a labeling that makes all nodes accept. If $P$ does not hold, in all labelings at least one node rejects. The cost of a scheme is its maximum label size. 

The \emph{Tradeoff Conjecture} [Feuilloley, Fraigniaud, Hirvonen, Paz, and Perry, DISC 18, Dist. Comput.~21]  hypothesizes that the existence of a $1$-PLS for a property $P$ with cost $p$  implies the existence of a $t$-PLS for $P$ with cost $O(\lceil p/t \rceil)$. The conjecture was initially shown to hold for specific graph classes, such as trees, cycles, and grids.
Later, a weaker $\widetilde{O}(\lceil \Delta p/\sqrt{t} \rceil)$ cost was shown for fixed minor-free graphs, where $\Delta$ is the maximum degree.

In this work we resolve the Tradeoff Conjecture, up to a single logarithmic factor. In general graphs, we show that the existence of a $1$-PLS with cost $p$ implies the existence of an $O(t\log{n})$-PLS with cost $O(\lceil p/t \rceil)$ for the same property. For fixed minor-free graphs (which include e.g. planar graphs), we show that the existence of a $1$-PLS with cost $p$ implies the existence of a $t$-PLS with cost $O(\lceil p/t \rceil+\log{n})$ for the same property. We also refute a previously suggested stronger variant of the Tradeoff Conjecture, and show that having very large $t$-hop neighborhoods is an insufficient condition for obtaining a tradeoff better than $O(\lceil p/t \rceil)$.    
\end{abstract}

	\vfill
	\begin{multicols}{2}
		{
   \small 
			\setcounter{secnumdepth}{5}
			\setcounter{tocdepth}{2} \tableofcontents
		}
	\end{multicols}